\documentclass[11pt,a4paper]{amsart}
\usepackage[foot]{amsaddr}
\usepackage{amsmath}
\usepackage{amsfonts}
\usepackage{amssymb}
\usepackage{amsthm}
\usepackage{algpseudocode}
\usepackage{fullpage}
\usepackage{microtype}
\usepackage{newtxmath}
\usepackage[tt=false]{libertine} 
\usepackage{caption}
\usepackage{tcolorbox}
\usepackage{bbm}
\usepackage{hyperref, color}
\hypersetup{colorlinks=true,citecolor=blue, linkcolor=blue, urlcolor=blue}
\usepackage[linesnumbered,boxed,ruled,vlined]{algorithm2e}
\usepackage{bm}
\usepackage{bbm}
\usepackage[numbers]{natbib}
\usepackage{xcolor}
\usepackage{enumerate} 
\usepackage{enumitem}
\usepackage{ragged2e}
\usepackage{tabularx}
\usepackage{array}
\usepackage{tablefootnote}
\usepackage{longtable}
\usepackage{hhline}
\usepackage{multirow}
\usepackage{cleveref}
\newcolumntype{L}[1]{>{\raggedright\arraybackslash}p{#1}}
\newcolumntype{C}[1]{>{\centering\arraybackslash}m{#1}}
\newcolumntype{R}[1]{>{\raggedleft\arraybackslash}p{#1}}

\usepackage{makecell}
\usepackage{footnote}
\usepackage{float}
\makesavenoteenv{tabular}

\renewcommand{\epsilon}{\varepsilon}

\newcommand{\ifarxiv}[2]{#2}
\renewcommand{\ifarxiv}[2]{#1}

\newtheorem{theorem}{Theorem}[section]

\newtheorem*{claim*}{Claim}
\newtheorem{condition}[theorem]{Condition}

\newtheorem{lemma}[theorem]{Lemma}
\newtheorem{proposition}[theorem]{Proposition}
\newtheorem{corollary}[theorem]{Corollary}
\theoremstyle{definition}

\newtheorem{definition}[theorem]{Definition}
\newtheorem{remark}[theorem]{Remark}
\newtheorem*{remark*}{Remark}
\newtheorem{assumption}{Assumption}

\renewcommand{\emptyset}{\varnothing}
\newcommand{\E}[1]{\mathbb{E}\left[{#1}\right]}
\renewcommand{\Pr}[2][]{ \ifthenelse{\isempty{#1}}
  {\mathop{\mathbf{Pr}}\left[#2\right]} {\mathop{\mathbf{Pr}}_{#1}\left[#2\right]} }

\newcommand{\one}[1]{\mathbbm{1}\left[#1\right]}

\newcommand{\abs}[1]{\left\vert#1\right\vert}

\newcommand{\tuple}[1]{\left(#1\right)}

\newcommand{\eps}{\varepsilon}
\newcommand{\tp}{\tuple}

\newcommand{\defeq}{:=}

\newcommand{\lb}{C}

\newcommand{\DTV}[2]{d_{\mathrm{TV}}\left({#1},{#2}\right)}

\def\*#1{\bm{#1}} 
\def\+#1{\mathcal{#1}} 
\def\-#1{\mathrm{#1}} 
\def\=#1{\mathbb{#1}} 

\usepackage[textsize=tiny]{todonotes}

\usepackage{xifthen}

\makeatletter
\def\prob#1#2#3{\goodbreak\begin{list}{}{\labelwidth\z@ \itemindent-\leftmargin
                        \itemsep\z@  \topsep6\p@\@plus6\p@
                        \let\makelabel\descriptionlabel}
                \item[\it Name]#1
               \item[\it Instance]                #2
                \item[\it Output]#3
                \end{list}}
\makeatother

\newcommand{\linbf}{\textnormal{\textsf{LinearBF}}}
\newcommand{\berdiv}{\textnormal{\textsf{BernoulliDivision}}}
\newcommand{\subbf}{\textnormal{\textsf{SubtractBF}}}
\newcommand{\poly}{{\rm poly}}  
\newcommand{\True}{\mathtt{True}}

\newcommand{\resolve}{\textnormal{\textsf{Resolve}}}

\newcommand{\sample}{\mathsf{Sample}}

\newcommand{\pred}{\textnormal{\textsf{pred}}}

\newcommand{\lsample}{\textnormal{\textsf{LocalSample}}}

\newcommand{\eval}{\textnormal{\textsf{Evaluate}}}

\renewcommand{\defeq}{\triangleq}
\usepackage{enumitem}

\newcommand{\chk}{\mathsf{Check}}
\newcommand{\test}{\mathsf{Test}}
\newcommand{\filter}{\mathsf{Filter}}

\newbool{doubleblind}
\setbool{doubleblind}{false}

\title{Local Gibbs sampling beyond local uniformity}

\date{}

\ifarxiv{
\author{Hongyang Liu, Chunyang Wang, Yitong Yin}
\address{State Key Laboratory for Novel Software Technology, New Cornerstone Science Laboratory, Nanjing University, 163 Xianlin Avenue, Nanjing, Jiangsu Province, 210023, China. \textnormal{E-mail: \texttt{liuhongyang@smail.nju.edu.cn}, \texttt{wcysai@smail.nju.edu.cn}, \texttt{yinyt@nju.edu.cn}}}
}
{
  \author{Author(s)}
}

\begin{document}

\begin{abstract}
Local samplers are algorithms that generate random samples based on local queries to high-dimensional distributions, ensuring the samples follow the correct induced distributions while maintaining time complexity that scales locally with the query size. 
These samplers have broad applications, including deterministic approximate counting~\cite{HWY22c,feng2023towards}, sampling from infinite or high-dimensional Gibbs distributions~\cite{AJ22,HWY22a}, and providing local access to large random objects~\cite{amartya2020local}.

In this work, we present local samplers for Gibbs distributions of spin systems.
%
Specifically, we design linear-time local samplers for:
\begin{itemize}
    \item spin systems with soft constraints, including the first local sampler for near-critical Ising models;
    \item truly repulsive spin systems, represented by the first local sampler for uniform proper $q$-colorings, with $q=O(\Delta)$ colors on graphs with maximum degree $\Delta$.
\end{itemize}
These local samplers are efficient beyond the ``local uniformity'' threshold, which imposes unconditional marginal lower bounds --- a key assumption required by all prior local samplers.
Our results show that, in general, local sampling is not significantly harder than global sampling for spin systems. As an application, our results also imply local algorithms for probabilistic inference in the same near-critical regimes. 
\end{abstract}	

\maketitle

\section{Introduction}

Spin systems, which originated in statistical physics, are stochastic models characterized by local interactions.
These models have not only advanced our understanding of physical phenomena, such as phase transitions and criticality, 
but have also become central to machine learning and theoretical computer science, 
particularly in the study of sampling and inference problems in complex distributions.

Let $q\ge 2$ be an integer. 
A $q$-spin system $\+S=(G=(V,E),\bm{\lambda}=(\lambda_v)_{v\in V},\bm{A}=(A_e)_{e\in E})$ is defined on a finite graph $G = ( V, E)$,
where each vertex is associated with an \emph{external field} $\lambda_v\in\mathbb{R}^q_{\geq 0}$, and each edge $e\in E$  is associated with an \emph{interaction matrix} $A_{e}\in\mathbb{R}^{q\times q}_{\geq 0}$. 
A configuration $\sigma \in [q]^V$ assigns a spin state from $[q]$ to each vertex $v \in V$. 

The \emph{Gibbs distribution} $\mu=\mu^{\+S}$ over all configurations $\sigma \in [q]^V$ is given by:
\[
\mu(\sigma)\defeq \frac{w(\sigma)}{Z}, \qquad w(\sigma)\defeq\prod\limits_{v\in V}\lambda_v(\sigma(v))\prod\limits_{e=(u,v)\in E}A_e(\sigma(u),\sigma(v)),
\]
where the the normalizing factor  $Z=\sum_{\sigma\in[q]^V}w(\sigma)$ is the \emph{partition function}.


A central question in the study of spin systems is sampling from their associated Gibbs distributions. For well-known models such as the hardcore model and the Ising model, a critical threshold determined by the system's parameters has been identified, beyond which sampling from the Gibbs distribution becomes \textbf{NP}-hard~\cite{SS14,galanis2016inapproximability}. Recent breakthroughs have shown that the Glauber dynamics, a widely-used Markov chain, mixes rapidly up to these critical thresholds for specific spin systems, including the hardcore model and the Ising model~\cite{ALO20, chen2020rapid, CLV21, chen2021rapid, anari2022entropic, CE22,CFYZ22}. These results provide a comprehensive characterization of the computational phase transition inherent in sampling from the Gibbs distributions of such spin systems.

\subsection{Local sampling and local uniformity}
Recent research has increasingly focused on \emph{local} sampling techniques for high-dimensional Gibbs distributions~\cite{AJ22,anand23sphere,feng2023towards}. Rather than directly drawing a global sample from the Gibbs distribution $\mu$, such algorithms aim to answer on-demand \emph{local} queries on a small subset of vertices $\Lambda\subseteq V$, and returns a sample approximately distributed according to the marginal distribution of $\mu$ induced on $\Lambda$, at a \emph{local} cost that depends only on the size of the query set $\abs{\Lambda}$ (and not on the total size $|V|$ of the spin system).
For a subset of vertices $\Lambda\subseteq V$, the marginal distribution $\mu_{\Lambda}$ is defined as:
\[
\forall \tau\in[q]^\Lambda,\quad \mu_\Lambda(\tau)\triangleq\sum_{\sigma\in[q]^V:\sigma_\Lambda=\tau}\mu(\sigma).
\]
Then, the local sampling problem is defined as follows:

\begin{center}
  \begin{tcolorbox}[=sharpish corners, colback=white, width=1\linewidth]
    \begin{center}
  \textbf{The local sampling problem}
    \end{center}
    \vspace{6pt}
\textbf{Input}: A spin system $\+S=(G,\bm{\lambda},\bm{A})$, where $G=(V,E)$, and a subset of vertices $\Lambda\subseteq V$;

\textbf{Goal}: Generate a sample $X\sim \mu_{\Lambda}$ in time that scales near-linearly in $|\Lambda|$.

  \end{tcolorbox} 
\end{center}
Local sampling can, of course, solve global sampling by simply querying all vertices or using the auto-regressive sampler for self-reducible problems, as in~\cite{AJ22, HWY22a}. 
Beyond this, these local samplers offer the ability to ``scale down'' the sampling process, 
addressing the challenge of providing local access to large random objects~\cite{amartya2020local}, 
where sublinear computational costs are required for sublinear-size queries.
For applications, efficient local samplers directly imply efficient algorithms for probabilistic inference for self-reducible problems, and can possibly lead to efficient approximate counting algorithms~\cite{HWY22c, feng2023towards, anand24approximate,anand2025sinkfree}. 

However, existing local samplers for spin systems~\cite{AJ22,anand23sphere,feng2023towards} rely on the assumption of \emph{unconditional marginal lower bounds}, also known as the ``\emph{local uniformity}'' property.
This assumption requires that the marginal distribution of each vertex remains nearly identical across all neighboring configurations, which may be excessively restrictive for various problems.

Consider, for example, the Ising model with edge activity $\beta>0$  and an arbitrary external field. 
The Ising model, introduced by Ising and Lenz~\cite{Ising1925Beitrag}, has been extensively studied in various fields.
Formally, it is a $2$-spin system with $A_e=\begin{pmatrix}
    \beta & 1\\
    1 & \beta\\
\end{pmatrix}$ at each edge $e\in E$ and arbitrary $\lambda_v$ at each $v\in V$.
Sampling from its Gibbs distribution can be achieved through Glauber dynamics, which mixes rapidly under the well-known ``uniqueness condition'':
\begin{equation}\label{eq:Ising-uniqueness-condition}
    \beta \in \left(\frac{\Delta-2}{\Delta},\frac{\Delta}{\Delta-2}\right),
\end{equation} 
where $\Delta$ is the maximum degree of the underlying graph.  Beyond this, either Glauber dynamics becomes torpidly mixing, or the sampling problem itself becomes intractable.
 In contrast, the requirement of unconditional marginal lower bounds for such models imposes a significantly stricter condition:
\begin{align}\label{eq:Ising-local-uniformity}
 \beta\in \left(\left(\frac{\Delta-1}{\Delta+1}\right)^{1/\Delta},\left(\frac{\Delta+1}{\Delta-1}\right)^{1/\Delta}\right)=\left(1-\frac{1}{\Theta(\Delta^2)},1+\frac{1}{\Theta(\Delta^2)}\right).   
\end{align}

Much greater challenges arise in “truly repulsive” spin systems --- most notably, in sampling uniform proper $q$-colorings. Given a graph $G=(V,E)$, a proper $q$-coloring is an assignment $\sigma:V\to [q]$ such that $\sigma(u)\neq \sigma(v)$ for all $(u,v)\in E$. 
This is one of the most extensively studied sampling problems. 
(Global) sampling algorithms have gradually lowered the tractability threshold for proper $q$-colorings to $q>1.809\Delta$~\cite{vigoda1999improved,chen2019improved,carlson2025flip}, where $\Delta$ is the maximum degree of the graph, while the uniqueness condition for proper $q$-colorings is given by $q\ge \Delta+1$.
On the other hand, the truly repulsive nature of proper colorings precludes local uniformity:
the marginal probability of a color at a vertex can drop to zero when a neighbor is assigned that color, so any method that relies on an unconditional lower bound on the marginals fails.
Consequently, to this day, no local sampler is known for uniform proper $q$-colorings.



This stark discrepancy raises a fundamental question: Do local samplers exist for such models in near-critical regimes? Or does local sampling inherently require a significantly more stringent critical condition compared to global sampling? 

\subsection{Our results}
In this paper, we address the aforementioned open question by designing new linear-time local samplers for two fundamental classes of spin systems under near-critical conditions: models with soft constraints, including the Ising model, and repulsive models, represented by proper $q$-colorings, showing that local sampling remains feasible near the global threshold for these models.

Our main contributions, both the first of their kind, are:
\begin{itemize}
    \item a local sampler for the Ising model in near-critical regimes;
    \item a local sampler for uniform proper $q$-colorings using $q=O(\Delta)$ colors.
\end{itemize}





Specifically, our local samplers assume the following natural access model for spin systems.

\begin{assumption}[probe access]\label{assumption:access-model}
Let $\+S=(G=(V,E),\bm{\lambda},\bm{A})$ be a $q$-spin system.
We assume:
\begin{itemize}
       \item For each $v\in V$, each neighbor $u\in N(v)$ can be accessed in $O(1)$ time.
    \item Each entry in every $\lambda_v$ and $A_e$ can be retrieved in $O(1)$ time.
\end{itemize}
These can be achieved by storing $G$ as an adjacency list and representing $\lambda_v$ and $A_e$ as arrays.
\end{assumption}

\subsubsection{Local sampler for spin systems with soft constraints}
Our first general result provides a linear-time local sampler for $q$-spin systems that satisfy the following sufficient condition.

\begin{condition}[tractable regime for spin systems with soft constraints]\label{cond:main}
Let $\delta>0$ be a parameter, and $\+S=(G,\bm{\lambda},\bm{A})$ be a $q$-spin system on a graph $G=(V,E)$ with maximum degree $\Delta\geq 1$.
The following condition holds:
\begin{itemize}
 \item (\textbf{Normalized}) All $\lambda_v$ and $A_e$ are normalized, i.e.,
\[
\forall v\in V,\quad \sum\limits_{c\in [q]}\lambda_v(c)=1\quad\text{ and }\quad\forall e\in E,\quad \max\limits_{i,j\in [q]}A_e(i,j)=1.
\]
 This normalization can be enforced without altering the Gibbs distribution.
 \item (\textbf{Soft constraints}) For every edge $e=(u,v)\in E$ and every pair of spin values $c_1,c_2\in [q]$,
    \[
    A_e(c_1,c_2) \geq \lb(\Delta,\delta)\defeq 1 - \frac{1-\delta}{2\Delta}.
    \]
\end{itemize}
\end{condition}

The following theorem presents our local sampler for spin systems with soft constraints. 
\begin{theorem}[local sampler for spin systems with soft constraints]\label{theorem:local-sampler}
There exists an algorithm that, given access (as in \Cref{assumption:access-model}) to a $q$-spin system $\+S=(G,\bm{\lambda},\bm{A})$ satisfying \Cref{cond:main}, with Gibbs distribution $\mu=\mu^{\+S}$,
and given a subset of vertices $\Lambda\subseteq V$, outputs a perfect sample $X\sim \mu_{\Lambda}$ in expected time $O\left(\Delta\log q \cdot |\Lambda|\right)$. 
\end{theorem}

The local sampler in \Cref{theorem:local-sampler} is perfect and terminates in time linear in $|\Lambda|$ in expectation.

Next, we apply \Cref{theorem:local-sampler} to one of the most important spin systems with soft constraints: the Ising model. 
Recall the definition of the Ising model, which is a 2-spin system with an interaction matrix 
$$A_e=\begin{pmatrix}
    \beta & 1\\
    1 & \beta\\
\end{pmatrix}$$ 
at each edge $e\in E$ and an arbitrary external field $\lambda_v$ at each vertex $v\in V$.
Note that this standard definition of the Ising model does not satisfy the normalization condition in \Cref{cond:main} when $\beta>1$. 
However, we can transform such a (ferromagnetic) Ising model to satisfy this normalization condition,
by using $A_e/\beta$ as the interaction matrix, without altering its Gibbs distribution. 

Applying \Cref{theorem:local-sampler} gives the following corollary, where \Cref{assumption:access-model} is implicitly assumed.

\begin{corollary}[local Ising sampler]\label{cor:ising}
There exists an algorithm that, given a Ising model with Gibbs distribution $\mu$ on a graph $G=(V,E)$ with maximum degree $\Delta\geq 1$, arbitrary external fields $\lambda_v$ at each $v\in V$, and edge activity $\beta$ satisfying \begin{equation}\label{eq:Ising-condition}
    \beta \in\left(\frac{\Delta-0.5}{\Delta},\frac{\Delta}{\Delta-0.5}\right),
     \end{equation}
and given a subset of vertices $\Lambda\subseteq V$, outputs a perfect sample $X\sim \mu_{\Lambda}$ in expected time $O\left(\Delta\cdot |\Lambda|\right)$. 
\end{corollary}
The condition in~\eqref{eq:Ising-condition} falls within the same regime of $\left(1-\Theta\left(\frac{1}{\Delta}\right),1+\Theta\left(\frac{1}{\Delta}\right)\right)$ as the uniqueness condition in \eqref{eq:Ising-uniqueness-condition}, substantially improving upon the local uniformity condition in \eqref{eq:Ising-local-uniformity}.

\subsubsection{Local sampler for proper $q$-colorings}

Our next result establishes a linear-time local sampler for one of the most fundamental repulsive spin systems: the uniform proper $q$-coloring model. Given a graph $G = (V, E)$, a proper $q$-coloring is an assignment $\sigma: V \to [q]$ such that $\sigma(u) \neq \sigma(v)$ for every edge $(u, v) \in E$.  This classical combinatorial model can be viewed as a $q$-spin system with truly repulsive hard constraints: the interaction matrix assigns zero weight to configurations where adjacent vertices share the same color (i.e., zeros on the diagonal) and unit weight otherwise (i.e., ones off the diagonal). 

Uniform sampling of proper $q$-colorings has long been a central problem in the study of algorithmic sampling and counting. In a seminal work, Jerrum~\cite{jerrum1995simple} established optimal mixing of the Glauber dynamics for proper $q$-colorings under the condition $q > 2\Delta$, where $\Delta$ denotes the maximum degree of the graph. This threshold was later improved by Vigoda~\cite{vigoda1999improved}, who showed that the flip dynamics mixes in $O(n \log n)$ time when $q > \frac{11}{6} \Delta$, which in turn implied an $O(n^2)$ mixing time for the standard Glauber dynamics. More recently, the threshold has been further lowered to $q > 1.809\Delta$ through a sequence of advances~\cite{chen2019improved, carlson2025flip}.

Despite this progress, efficient local samplers have remained elusive for proper $q$-colorings, primarily due to the lack of unconditional marginal bounds, as discussed earlier. 
Previously, as noted in \cite[Section 8]{feng2023towards}, a major obstacle to designing a local sampler for $q$-colorings has been overcoming the threshold $q=\Omega(\Delta^2)$, which corresponds to Huber's bounding chains~\cite{huber1998exact}.

Our result overcomes this obstacle and provides the first local sampling algorithm for proper $q$-colorings in the near-critical regime $q = O(\Delta)$.

\begin{theorem}[local sampler for proper $q$-colorings]\label{theorem:coloring}
There exists an algorithm that, given a graph $G=(V,E)$ with maximum degree $\Delta\geq 1$,  
an integer $q$ satisfying
\begin{equation}\label{eq:coloring-condition}
   q\geq 65\Delta,
     \end{equation}
and a subset of vertices $\Lambda\subseteq V$, 
 outputs a perfect sample $X\sim \mu_{\Lambda}$, where $\mu$ denotes the uniform distribution over all proper $q$-colorings of $G$, in expected time $O\left(\Delta^2q\cdot |\Lambda|\right)$.  
\end{theorem}

\begin{remark}
%
If the local computation cost is relaxed to be sublinear in the size of the input graph,  
as in the local computation algorithm (LCA) model, 
a better bound of $q\geq 9\Delta$ was obtained in~\cite{amartya2020local}. 
    Specifically, given any subset of vertices $\Lambda\subseteq V$, their algorithm outputs an approximate sample $X$ from $\mu_{\Lambda}$ within $\varepsilon$ total variation distance, in time $\tilde{O}((|V|/\varepsilon)^{0.68} \Delta \abs{\Lambda})$.  
    Their approach is based on simulating distributed local Markov chains and is unlikely to yield a local sampler in the sense of \Cref{theorem:coloring}.
\end{remark}

\subsubsection{Local algorithms for probabilistic inference}

An important application of efficient local samplers lies in their connection to local counting for self-reducible problems, where they directly yield efficient algorithms for probabilistic inference. In the (Bayesian) probabilistic inference problem, the goal is typically to estimate how the marginal probability of a specific vertex changes under certain conditions or observations. This task is fundamental to many areas and is particularly well-motivated in machine learning and statistics, where inference plays a central role in prediction, decision-making, and learning~\cite{dagum1993approximate, dagum1997optimal}.

For a partial configuration $\sigma\in [q]^{\Lambda}$ over a subset of vertices $\Lambda\subset V$ with $\mu_{\Lambda}(\sigma)>0$, and a vertex $v\in V\setminus \Lambda$, the conditional marginal distribution $\mu^{\sigma}_v$ is defined as:
\[
\forall c\in [q],\quad \mu^{\sigma}_v(c)\triangleq \frac{\sum_{\tau\in[q]^V:\tau_\Lambda=\sigma,\tau_v=c}\mu(\tau)}{\mu_{\Lambda}(\sigma)}=\Pr[\tau\sim{\mu}]{\tau_v=c\mid \tau_\Lambda=\sigma}.
\]

Specifically, we obtain the following local algorithms for probabilistic inference in spin systems with soft constraints and for proper $q$-colorings.

\begin{theorem}[probabilistic inference in spin systems with soft constraints]\label{theorem:inference-permissive}
There exists an algorithm that, 
given access (as in \Cref{assumption:access-model}) to a $q$-spin system $\+S=(G,\bm{\lambda},\bm{A})$ with Gibbs distribution $\mu=\mu^{\+S}$ satisfying \Cref{cond:main}, 
and given a subset of vertices $\Lambda\subset V$, a partial configuration $\sigma\in [q]^{\Lambda}$ with $\mu_{\Lambda}(\sigma)>0$, a vertex $v\in V\setminus \Lambda$, and parameters $\varepsilon,\delta\in (0,1)$, 
outputs an estimate $\hat{\mu}^{\sigma}_v$ such that 
\[
\Pr{\forall c\in [q]: (1-\varepsilon)\mu^{\sigma}_v(c)\leq \hat{\mu}^{\sigma}_v(c)\leq (1+\varepsilon)\mu^{\sigma}_v(c)}\ge1-\delta,
\]
in expected time $O\left(\varepsilon^{-2}\delta^{-1}\Delta q^2\log q\cdot |\Lambda|\right)$. 
\end{theorem}

\begin{theorem}[probabilistic inference for proper $q$-colorings]\label{theorem:inference-coloring}
There exists an algorithm that, 
given a graph $G=(V,E)$ with maximum degree $\Delta$ and $q\ge 65\Delta$,
a partial proper $q$-coloring $\sigma\in [q]^{\Lambda}$ of a subset of vertices $\Lambda\subset V$, a vertex $v\in V\setminus \Lambda$, and parameters $\varepsilon,\delta\in (0,1)$, 
outputs an estimate $\hat{\mu}^{\sigma}_v$ such that 
\[
\Pr{\forall c\in [q]: (1-\varepsilon)\mu^{\sigma}_v(c)\leq \hat{\mu}^{\sigma}_v(c)\leq (1+\varepsilon)\mu^{\sigma}_v(c)}\ge 1-\delta,
\] 
where $\mu$ denotes the uniform distribution over all proper $q$-colorings of $G$,
in expected time $O\left(\varepsilon^{-2}\delta^{-1}\Delta^2 q^3|\Lambda|\right)$. 
\end{theorem}

\subsection{Technique overview}
Previous works on local samplers include \cite{AJ22} and \cite{feng2023towards}, both of which rely on unconditional marginal lower bounds, i.e., the local uniformity property.
The work of \cite{AJ22} introduced a novel local sampler called ``\emph{lazy depth-first search}'' ({a.k.a.}~the A-J algorithm). 
To sample the spin of a vertex according to its correct marginal distribution, the algorithm first draws a random spin according to the unconditional marginal lower bounds, and with the remaining probability, it recursively samples the spins of all neighboring vertices.
The algorithm in \cite{feng2023towards} takes a different approach, employing a backward deduction framework for Markov chains, referred to as ``\emph{coupling towards the past}'' (CTTP). 
Their method uses systematic Glauber dynamics combined with a grand coupling based on unconditional marginal lower bounds, 
allowing the spin of a vertex to be inferred via a convergent information-percolation process.
Despite their differences, both approaches rely crucially on unconditional marginal lower bounds (implied by local uniformity) to prevent excessive backtracking and thus ensure the efficiency of the sampling procedure.
For a more detailed comparison of the two algorithms, we refer the reader to \cite[Section 1.2]{feng2023towards}.

We introduce key innovations that eliminate the reliance on local uniformity for local sampling.
While the high-level ideas are broadly applicable, 
we present our new local samplers within the coupling towards the past (CTTP) framework for local Markov chains.
Unlike the original CTTP algorithm of \cite{feng2023towards}, which depends on a default grand coupling derived from unconditional marginal lower bounds,
our approach introduces several new ideas to adapt the grand coupling, enabling efficient local samplers without assuming unconditional marginal lower bounds.

To design local samplers for systems with soft constraints that lack local uniformity, we first generalize the CTTP framework via an abstract notion of  \emph{marginal sampling oracles}: 
procedures that sample from conditional marginal distributions given oracle access to the neighborhood configuration.
This abstraction allows each implementation of a marginal sampling oracle to correspond to a specific simulation of Glauber dynamics --- or more precisely, to a particular grand coupling of the chain.
We then implement the marginal sampling oracle via rejection sampling, 
which leverages the softness of local constraints rather than relying on unconditional lower bounds on marginal probabilities.
This yields efficient local samplers for spin systems with soft constraints beyond local uniformity.

The case of truly repulsive spin systems is much more challenging, as no marginal lower bound exists. 
Consequently, it is impossible to determine the outcome of an update at a given time with positive probability without additional information.
To address this, we further extend the CTTP framework to:
\begin{itemize}
  \item allow \emph{partial information} (rather than the full outcome) to be resolved at a given timestamp;
  \item allow the grand coupling strategy at timestamp $t$ to depend on earlier timestamps $t'<t$,
  introducing \emph{adaptivity} into the grand couplings.
\end{itemize}
Leveraging these new ideas, we obtain the first local sampler for $q$-colorings with $q=O(\Delta)$ colors.
We note that similar ideas have appeared in Coupling From The Past (CFTP), which yields (global)  perfect samplers for $q$-colorings with $q=O(\Delta)$ colors~\cite{huber1998exact,bhandari2020improved,jain2021perfectly}.
Our technical contributions regarding $q$-colorings can thus be viewed as local counterparts of these CFTP-based global samplers.

For the analyses of our local samplers, correctness follows from the validity of the underlying grand coupling in each construction.
For efficiency, we employ different approaches for spin systems with soft and hard constraints.
In the case of spin systems with soft constraints, the algorithm’s behavior is relatively straightforward: we demonstrate that it is stochastically dominated by a subcritical branching process, which directly implies its efficiency.
%
In contrast, the $q$-coloring case exhibits more intricate behavior, rendering the previous analysis inapplicable. 
To address this, we introduce a carefully designed \emph{potential function} that reflects the state of the algorithm and drops to zero upon termination.
We prove that this potential function evolves as a supermartingale with bounded differences throughout the execution of the algorithm, thereby establishing efficiency.

\subsection{Related topics}

Our local sampler is built upon the Coupling Towards The Past (CTTP) framework introduced in~\cite{feng2023towards}, which bears resemblance to the celebrated Coupling From The Past (CFTP) method by Propp and Wilson~\cite{propp1996exact} for perfect sampling from Markov chains, as both approaches utilize the idea of grand coupling. (See Section 1.4 of~\cite{feng2023towards} for a detailed comparison between the two frameworks.) Notably, the CTTP framework is more restrictive than CFTP: an efficient local sampler within the CTTP framework implies the existence of an efficient perfect sampler under CFTP, but the converse does not hold. This asymmetry arises because CTTP aims to produce not only a perfect sample but also a \emph{local} one, whereas existing CFTP constructions typically rely on global knowledge in the analysis~\cite{huber1998exact,bhandari2020improved,jain2021perfectly}.

The backward deduction of Markov chain states in the CTTP framework also bears resemblance to the analysis of the cutoff phenomenon via the method of \emph{information percolation}~\cite{Lubetzky2016information,Lubetzky2017universal}. In particular, \cite{Lubetzky2017universal} shows that Glauber dynamics for the ferromagnetic Ising model exhibits a cutoff phenomenon in the near-critical regime $\beta \leq 1 + \frac{1}{O(\Delta)}$. Despite these structural similarities, the goals of the two frameworks differ fundamentally: CTTP is designed for constructing local samplers, while the information percolation approach is aimed at analyzing mixing times. Furthermore, our technique for obtaining near-critical local samplers for the Ising model differs significantly from that of~\cite{Lubetzky2017universal}: our grand coupling at each time step is constructed using rejection sampling, whereas theirs is based on discrete Fourier expansion. Additionally, the bounds we obtain are tighter than those in~\cite{Lubetzky2017universal}.


Our local sampler also falls into the category of providing \emph{local access to large random objects}~\cite{amartya2020local,biswas2022local,morters2022sublinear}.
Given a $q$-spin system $\+S = (G = (V, E), \bm{\lambda}, \bm{A})$ and public random bits, our algorithm can generate consistent samples $X_{\Lambda}$ such that $X \sim \mu = \mu^{\+S}$ upon multiple queries of any subset of vertices $\Lambda \subseteq V$, using only a local number of probes for public random bits.


\subsection{Organization} 
The paper is organized as follows:
\begin{itemize}
    \item In \Cref{sec:prelim}, we introduce the necessary preliminaries.
    \item 
    In \Cref{sec:framework}, we present a generalized  CTTP framework, with an abstract notion of ``marginal sampling oracles'', and show how to utilize this abstraction to yield local samplers beyond local uniformity.
    \item 
    In \Cref{sec:local}, we design a new marginal sampling oracle and apply it to obtain our local sampler for spin systems with soft constraints, proving \Cref{theorem:local-sampler,theorem:inference-permissive}.
    \item 
    In \Cref{sec:coloring}, we further extend the CTTP framework to design a local sampler for $q$-colorings, proving \Cref{theorem:coloring,theorem:inference-coloring}.
    \item 
    In \Cref{sec:conclusions}, we summarize our contributions and outline potential future directions.
\end{itemize}


\section{Preliminaries}\label{sec:prelim}

\subsection{Markov chain basics}

Let $\Omega$ be a (finite) state space.
Let $(X_t)_{t = 1}^\infty$ be a Markov chain over the state space $\Omega$ with transition matrix $P$.
A distribution $\pi$ over $\Omega$ is a \emph{stationary distribution} of $P$ if $\pi = \pi P$.
The Markov chain $P$ is \emph{irreducible} if for any $x,y \in \Omega$, there exists a timestamp $t$ such that $P^t(x,y) > 0$.
The Markov chain $P$ is \emph{aperiodic} if for any $x \in \Omega$, $\gcd\{t\mid P^t(x,x) > 0\} = 1$.
If the Markov chain $P$ is both irreducible and aperiodic, then it has a unique stationary distribution.
The Markov chain $P$ is \emph{reversible} with respect to the distribution $\pi$ if the following \emph{detailed balance equation} holds.
\begin{align*}
	\forall x, y \in \Omega,\quad \pi(x) P(x,y) = \pi(y)P(y,x),
\end{align*}
which implies $\pi$ is a stationary distribution of $P$.
The \emph{mixing time} of the Markov chain $P$ is defined by
\begin{align*}
	\forall \epsilon > 0, \quad T(P,\epsilon) \defeq \max_{X_0 \in \Omega} \max\{t \mid \DTV{P^t(X_0,\cdot)}{\pi} \leq \epsilon\},
\end{align*}
where the \emph{total variation distance} is defined by
\begin{align*}
\DTV{P^t(X_0,\cdot)}{\pi} \defeq \frac{1}{2}\sum_{y \in \Omega}\abs{P^t(X_0,y)-\pi(y)}.	
\end{align*}

\subsection{Systematic scan Glauber dynamics}\label{sec:sys-scan}

The \emph{systematic scan} Glauber dynamics is a generic way to sample from Gibbs distributions defined by spin systems. Given a $q$-spin system $\+S=(G=(V,E),\bm{\lambda},\bm{A})$. Let $n=|V|$ and assume an arbitrary ordering $V=\{v_0,v_1,\dots,v_{n-1}\}$, and let $T>0$ be some finite integer, the $T$-step systematic scan Glauber dynamics $\+P{(T)}=\+P^{\+S}(T)$
\begin{enumerate}
    \item starts with an arbitrary configuration $X_{-T}\in [q]^V$ satisfying $\mu(X_{-T})>0$ at time $t=-T$;
    \item at each time $-T< t\leq  0$,
    \begin{enumerate}
        \item picks the vertex $v= v_{i(t)}$ where $i(t)\triangleq t\mod n$, let $X_t(u)=X_{t-1}(u)$ for every $u\in V\setminus \{v\}$;
        \item resample $X_t(v)$ from the marginal distribution $\mu^{X_{t-1}}_{v}$ on $v$ conditioning on $X_{t-1}$ where
        \[
        \forall c\in [q],\quad \mu^{X_{t-1}}_{v}(c)=\mu^{X_{t-1}(N(v))}_{v}(c)\propto \lambda_v(c)\prod\limits_{e=(u,v)\in E}A_{e}(\sigma(u),c).
        \]
        Here, the first equality is due to the \emph{conditional independence} property of Gibbs distributions.
   \end{enumerate}
\end{enumerate}

The systematic scan Glauber dynamics is not a time-homogeneous Markov chain. However, by bundling $n$ consecutive updates together, we can obtain a time-homogeneous Markov chain, which is aperiodic and reversible, which is sufficient for us to apply the following theorem.
\begin{theorem}[\cite{levin2017markov}]\label{thm-convergence}
Let $\mu$ be a distribution with support $\Omega \subseteq [q]^V$. Let $(X_t)_{t =0}^\infty$ denote the systematic scan Glauber dynamics on $\mu$. If $(X_t)_{t =0}^\infty$ is irreducible over $\Omega$, it holds that
\begin{align*}
	\forall X_0 \in \Omega,\quad \lim_{t \to \infty}\DTV{X_t}{\mu} = 0.
\end{align*}
\end{theorem}

\section{Coupling towards the past without marginal lower bounds}\label{sec:framework}

Our local sampler is based on the Coupling Towards The Past (CTTP) framework recently introduced in \cite{feng2023towards}, 
which constructs a local sampler by evaluating multiple spin states from stationary Markov chains through backward deduction. 
Our framework generalizes the CTTP framework by replacing the default grand coupling, 
which uses unconditional marginal lower bounds, 
with grand couplings defined by arbitrary ``marginal sampling oracles''. 
This generalization allows us to design a specific marginal sampling oracle that leads to a local sampler beyond the regime of local uniformity.



\subsection{Marginal sampling oracles}
Before introducing the CTTP framework, we first define \emph{marginal sampling oracles}.
Due to the conditional independence property of Gibbs distributions, to sample from the (conditional) marginal distribution $\mu^{\sigma}_v$ for a vertex $v \in V$ and a configuration $\sigma \in [q]^{V \setminus {v}}$, it suffices to retrieve the spins of all neighbors, $\sigma(N(v))$.
A \emph{marginal sampling oracle} generalizes this concept by producing a marginal sample, given oracle access to the spin $\sigma(u)$ of each neighbor $u \in N(v)$.


\begin{definition}[marginal sampling oracle]\label{definition:locally-defined-grand-coupling}
Let $\mu$ be a distribution over $[q]^V$. For a variable $v \in V$, we define $\eval^\+O(v)$ as a procedure that
    makes oracle queries to $\+O(u)$, which consistently returns a value $c_u \in [q]$ for each $u \in N(v)$.

We say that $\eval^\+O(v)$ is a \emph{marginal sampling oracle} at $v$ (with respect to $\mu$) if: 
\begin{itemize}
    \item for each $\sigma \in [q]^{N(v)}$, assuming $\+O(u)$ consistently returns $\sigma(u)$ for each $u \in N(v)$, the output of $\eval^\+O(v)$ is distributed exactly as $\mu^{\sigma}_v$.
\end{itemize}
\end{definition}

Recall the definition of systematic scan Glauber dynamics $\+P(T)$ in \Cref{sec:sys-scan}.
Using a marginal sampling oracle, the systematic scan Glauber dynamics can be simulated as follows.

\begin{definition}[simulation of systematic scan Glauber dynamics via a marginal sampling oracle]\label{definition:systematic-scan-coupling-realization}
The systematic scan Glauber dynamics $\+P(T)$  with respect to $\mu$  is simulated as:
\begin{enumerate}
    \item start with an arbitrary configuration $X_{-T}\in [q]^V$ satisfying $\mu(X_{-T})>0$ at time $t=-T$;
    \item at each time $-T< t\leq  0$,
    \begin{enumerate}
        \item pick the vertex $v= v_{i(t)}$ where $i(t)\triangleq t\mod n$, let $X_t(u)=X_{t-1}(u)$ for every $u\in V\setminus \{v\}$;
        \item let $\eval^\+O(v)$ be a marginal sampling oracle (w.r.t $\mu$) at $v$ where the oracle accesses $\+O(u)$ are replaced with $X_{t-1}(u)$ for each $u\in N(v)$, update $X_t(v)\gets \eval^\+O(v)$.\label{item:update}
    \end{enumerate}
\end{enumerate}
\end{definition}

\begin{remark}[grand coupling]\label{remark:implicit-grand-coupling}
In \Cref{definition:systematic-scan-coupling-realization}, the only randomness involved is within the subroutine $\eval^{\+O}(v)$. 
Notably, for any implementation of a marginal sampling oracle, 
\Cref{definition:systematic-scan-coupling-realization} specifies a simulation of systematic scan Glauber dynamics,
and implicitly defines a \emph{grand coupling} that couples the Markov chain across all possible initial configurations. 
To see this, consider pre-sampling all random variables used within $\eval^\+O(v_{i(t)})$ for each timestamp $t$, thereby defining the grand coupling.
\end{remark}

\subsection{Simulating stationary Markov chains using backward deduction}

We present the CTTP framework for constructing the local sampler, 
which is a backward deduction of the forward simulation described in \Cref{definition:systematic-scan-coupling-realization} (or equivalently, the grand coupling constructed in \Cref{remark:implicit-grand-coupling}).

Consider the systematic scan Glauber dynamics running from the infinite past toward time $0$, which is the limiting process of $\+P(T)$ as $T\to \infty$, denoted by $\+P(\infty)$. 
By \Cref{thm-convergence}, when $\+P(T)$ is irreducible, the state $X_0$ of this process is distributed exactly according to $\mu$. Our local sampler is then constructed by resolving the outcome $X_0(\Lambda)$, where $\Lambda$ is the queried set of vertices. For any $t\leq 0$ and $u\in V$, define \begin{equation}\label{eq:definition-pred}
\pred_t(u)\defeq \max \{t'\mid t'\leq t, v_{i(t')}=u\}
\end{equation}
as the last time, up to $t$, that vertex $u$ was updated.   The local sampler is formally presented in \Cref{Alg:lsample}.

\begin{algorithm}[H]
\caption{$\lsample(\Lambda;M)$} \label{Alg:lsample}
\SetKwInput{KwData}{Global variables}
\KwIn{A $q$-spin system $\+S=(G=(V,E),\bm{\lambda},\bm{A})$, a subset of variables $\Lambda\subseteq V$.}
\KwOut{A random configuration $X \in [q]^{\Lambda}$.}
\KwData{A mapping $M: \mathbb{Z}\to [q]\cup\{\perp\}$.}
$X\gets \emptyset, M\gets \perp^{\mathbb{Z}}$\label{Line:local-sampler-initialization}\;
\ForAll{$v\in \Lambda$}{
    $X(v)\gets \resolve(\pred_0(v);M)$\;\label{Line:local-sampler-resolve}
}
\Return $X$\;
\end{algorithm}

\Cref{Line:local-sampler-resolve} of \Cref{Alg:lsample} utilizes a procedure $\resolve$, formally presented as \Cref{Alg:resolve}, 
which takes as input a timestamp $t\leq 0$, and determines the outcome of the update at time $t$ of $\+P(\infty)$.

\begin{algorithm}[H]
\caption{$\resolve(t; M)$} \label{Alg:resolve}
\SetKwInput{KwData}{Global variables}
\KwIn{A $q$-spin system $\+S=(G=(V,E),\bm{\lambda},\bm{A})$, a timestamp $t\leq 0$.}
\KwOut{A random value $x\in [q]$.}
\KwData{A mapping $M: \mathbb{Z}\to [q]\cup\{\perp\}$.}
\lIf(\tcp*[f]{check if the outcome is already resolved}){$M(t) \neq \perp$}{\Return $M(t)$\label{Line:resolve-finite-memoization}}
$M(t)\gets \eval^\+O(v_{i(t)})$, with $\+O(u)$ replaced by $\resolve(\pred_t(u);M)$ for each $u\in N(v_{i(t)})$\;\label{Line:resolve-sample}
\Return $M(t)$\; \label{Line:resolve-final-return}
\end{algorithm}

A global data structure $M$ is maintained within \Cref{Alg:lsample}, storing the resolved values $M(t)$ for updates at each time $t$.
It is initialized as $M=\perp^{\mathbb{Z}}$ in \Cref{Line:local-sampler-initialization}. 
This data structure $M$ is introduced to facilitate memoization:
the outcome of $\resolve(t)$ is evaluated only once, ensuring consistency across multiple calls for the same $t$. 
For simplicity, we omit explicit references to $M$ and write $\lsample(\Lambda)$ and $\resolve(t)$ instead of $\lsample(\Lambda;M)$ and $\resolve(t;M)$.

\Cref{Line:resolve-sample} of \Cref{Alg:resolve} invokes a procedure $\eval^{\+O}(v_{i(t)})$,
which is abstractly defined in \Cref{definition:locally-defined-grand-coupling} but is not yet fully implemented.
Recall that the purpose of $\eval^{\+O}(v_{i(t)})$ is to infer the value to which the vertex $v=v_{i(t)}$ is updated in $\+P(\infty)$ at time $t$, which is distributed as $\mu^{X_{t-1}(N(v))}_v$ given access to $X_{t-1}(N(v))$. 
However, since \Cref{Alg:resolve} implements a backward deduction (as opposed to a forward simulation)  of the chain, the neighborhood configuration $X_{t-1}(N(v))$ at time $t$ is not available directly. 
To address this, the algorithm recursively applies \Cref{Alg:resolve} to infer the last updated value of each neighbor $u\in N(v)$ before time $t$ (as specified in \Cref{Line:resolve-sample} of \Cref{Alg:resolve}).

Formally, the subroutine $\eval(t)$ must satisfy the following local correctness condition:

\begin{condition}[local correctness of $\eval^\+O(v)$]\label{condition:local-correctness}
For each $v\in V$, the procedure $\eval^{\+O}(v)$ is a marginal sampling oracle at $v$, satisfying the requirement of \Cref{definition:locally-defined-grand-coupling}.
\end{condition}

Recall that \Cref{Alg:resolve} is designed to resolve the outcome of $\+P(\infty)$ at time $0$. However, the limiting process $\+P(\infty)$ is well-defined only if $\+P(T)$ is irreducible.  Additionally, we note that \Cref{Alg:resolve} does not necessarily terminate. 
Nonetheless, we provide a sufficient condition that ensures both the irreducibility of $\+P(T)$ and the termination of \Cref{Alg:resolve}.

\begin{condition}[immediate termination of $\eval^{\+O}(v)$]\label{condition:immediate-termination}
   For each $v\in V$, let $\+E_v$ be the event that $\eval^{\+O}(v)$ terminates without making any calls to $\+O$. 
   Then, the following must hold:
        \[
        \Pr{\+E_v}>0.
        \]
\end{condition}


We now establish the correctness of \Cref{Alg:lsample}, assuming Conditions \ref{condition:local-correctness} and \ref{condition:immediate-termination}.
\begin{lemma}[conditional correctness of \Cref{Alg:lsample}]\label{lemma:lsample-correctness}
   Assume that Conditions \ref{condition:local-correctness} and \ref{condition:immediate-termination} hold for $\eval^{\+O}(v)$.
   Then, for any $\Lambda \subseteq V$, \Cref{Alg:lsample} terminates with probability 1 and returns a random value $X \in [q]^{\Lambda}$ distributed according to $\mu_{\Lambda}$ upon termination. 
\end{lemma}
\Cref{lemma:lsample-correctness} is proved later in \Cref{sec:appendix-correctness}.
 It guarantees the termination and correctness of \Cref{Alg:lsample}, without addressing its efficiency.
 Next, we provide a sufficient condition for the efficiency of \Cref{Alg:lsample}.

\begin{condition}[condition for fast termination of $\eval^{\+O}(v)$]\label{condition:fast-termination}
Let $\delta > 0$ be a parameter.
 For each $v\in V$, and each $\sigma\in [q]^{N(v)}$ such that $\+O(u)$ consistently returns $\sigma(u)$ for all $u\in N(v)$,
 let $\+{T}^{\sigma}_v$ denote the total number of calls to $\+O(u)$ over all $u\in N(v)$. 
 Then, the following holds:
    \[
    \E{\+{T}^{\sigma}_v}\leq 1-\delta.
    \]
\end{condition}


We conclude this subsection with the following lemma, which establishes the efficiency of our local sampler under the assumption of \Cref{condition:fast-termination}.
Notably, it provides an upper bound on the total number of recursive $\resolve$ calls, rather than simply counting the number of initial calls for each $t \leq 0$.

\begin{lemma}[conditional efficiency of \Cref{Alg:lsample}]\label{lemma:lsample-efficiency}
Assuming \Cref{condition:fast-termination} holds, the expected total number of calls to $\resolve(t)$  within $\lsample(\Lambda)$ is $O(|\Lambda|)$.
\end{lemma}

\begin{proof}
The introduction of the map $M$ in \Cref{Line:resolve-finite-memoization} of $\resolve(t)$ is for memoization and only reduces the number of recursive calls. As a result, the expected running time of $\lsample(\Lambda)$ can be upper-bounded by the sum of the expected running times of $\resolve(\pred_0(v))$ for each $v \in \Lambda$. It remains to show that the expected running time of $\resolve(\pred_0(v))$ is $O(1)$ for each $v \in V$.

  
As the mapping $M$ only reduces the number of recursive calls, the behavior of $\resolve(\pred_0(v))$ can be stochastically dominated by the following multitype Galton-Watson branching process: 
\begin{itemize} 
\item Start with a root node labeled with $\pred_0(v)$ at depth $0$. 
\item For each $i = 0, 1, \ldots$: for all current leaves labeled with some timestamp $t$ at depth $i$: 
\begin{itemize} 
    \item Perform an independent run of $\resolve(t)$, and for each timestamp $t' < t$ such that $\resolve(t)$ is directly recursively called, add a new node labeled with $t'$ as a child of $t$. 
\end{itemize} 
\end{itemize}
By \Cref{condition:fast-termination}, for any timestamp $t \leq 0$, the expected number of offspring of a node labeled $t$ is at most $1 - \delta$. Thus, applying the theory of branching processes, the expected number of nodes generated by this process is at most $\delta^{-1} = O(1)$. Therefore, the expected number of $\resolve(t)$ calls within $\lsample(\Lambda)$ is $O(|\Lambda|)$, completing the proof of the lemma. 
\end{proof}

\subsection{Conditional correctness of the local sampler}\label{sec:appendix-correctness}
We will prove \Cref{lemma:lsample-correctness}, which addresses the conditional correctness of the local sampler (\Cref{Alg:lsample}).
At a high level, the proof follows the same structure of the proof in~\cite{feng2023towards}.

First, we need to establish some basic components.
\begin{lemma}\label{lemma:lsample-correctness-irreducibility}
 Assume that both Conditions \ref{condition:local-correctness} and \ref{condition:immediate-termination} hold for $\eval^{\+O}(v)$, then $\+P(T)$ is irreducible.
\end{lemma}
\begin{proof}
    Recall that in \Cref{condition:immediate-termination}, for any $v\in V$ and $\sigma\in [q]^{N(v)}$, $\+E^{\sigma}_v$ denotes the event that $\eval^{\+O}(v)$ terminates without any calls to $\+O$, assuming that $\+O(u)$ consistently returns $\sigma(u)$ for each $u\in N(v)$. For any $v\in V$, let $c_v\in [q]$ be an arbitrary possible outcome of $\eval^{\+O}(v)$, conditioning on $\+E_v$ happens. Note that such $c_v$ always exists by \Cref{condition:immediate-termination}. Then combining with \Cref{condition:local-correctness}, we have 
\begin{equation}\label{eq:marginal-lower-bound}
   \min\limits_{\sigma \in [q]^{N(v)}}\mu^{\sigma}_v(c_v)>0. 
\end{equation}

Let $\tau\in [q]^V$ be the constant configuration where $\tau(v)=c_v$ for $t=\pred_0(v)$. By \eqref{eq:marginal-lower-bound} and the chain rule, we see that $\mu(\tau)>0$. Also following \eqref{eq:marginal-lower-bound}, any $\sigma\in [q]^V$ such that $\mu(\sigma)>0$ can reach $\tau$ through Glauber moves by changing some $\sigma(u)$ to $\tau(u)$ one at a time. Note that $\+P(T)$ is reversible, therefore, any $\sigma\in [q]^V$ such that $\mu(\sigma)>0$ can also be reached from $\tau$ and hence $\+P(T)$ is irreducible.
\end{proof}

For any finite $T>0$, we introduce the following finite-time version of \Cref{Alg:lsample}, presented as \Cref{Alg:lsample-finite}, which locally resolves the final state $X_0$ of $\+P(T)$. Note that the only difference between Algorithms \ref{Alg:lsample} and \ref{Alg:lsample-finite} is the different initialization of the map $M$.

\begin{algorithm}[H]
\caption{$\lsample_T(\Lambda)$} \label{Alg:lsample-finite}
\SetKwInput{KwData}{Global variables}
\KwIn{a $q$-spin system $\+S=(G=(V,E),\bm{\lambda},\bm{A})$, a subset of variables $\Lambda\subseteq V$}
\KwOut{A random configuration $X \in [q]^{\Lambda}$}
\KwData{a map $M: \mathbb{Z}\to [q]\cup\{\perp\}$}
$X\gets \emptyset, M(t)\gets X_{-T}(v_{i(t)}) \text{ for each } t\leq -T,  M(t)\gets \perp \text{ for each } t> -T$\;
\ForAll{$v\in \Lambda$}{
    $X(v)\gets \resolve(\pred_0(v))$\;
}
\Return $X$\;
\end{algorithm}
 We then have the following lemma.

\begin{lemma}\label{lemma:lsample-correctness-convergence}
     Assume that both Conditions \ref{condition:local-correctness} and \ref{condition:immediate-termination} hold for $\eval^{\+O}(v)$. Then,
     \begin{enumerate}
         \item $\lsample(\Lambda)$ terminates with probability $1$;\label{item:lsample-correctness-convergence-1}
         \item For any initial state $X_{-T}$, it holds that $\lim\limits_{T\to \infty}\DTV{\lsample(\Lambda)}{\lsample_T(\Lambda)}=0$.\label{item:lsample-correctness-convergence-2}
     \end{enumerate}
\end{lemma}

\begin{proof}
We start with proving \Cref{item:lsample-correctness-convergence-1}. It suffices to show the termination of $\resolve(t_0)$ for any $t_0\leq 0$. Recall the event $\+E_v$ in \Cref{condition:immediate-termination}. For each $t\leq 0$, we similarly let $\+E_t$ denote the event that $\eval^{\+O}(v_{i(t)})$ within $\resolve(t)$ terminates without making any calls to $\+O$. We also define the event:
\[
\+B_t: \+E_{t'} \text{ happens for all }t'\in [t-n+1,t].
\]
We claim that if $\+{B}_t$ happens for some $t\leq t_0$,  then no recursive calls to $\resolve(t')$ would be incurred for any $t'\leq t-n$ within $\lsample(\Lambda)$. For the sake of contradiction, assume that a maximum $t^*\leq t-n$ exists such that $\resolve(t^*)$ is called. As $t^*\leq t-n<t_0$ , $\resolve(t^*)$ must be recursively called directly within another instance of $\resolve(t')$ (through $\eval^{\+O}(v_{i(t')})$) such that $t^*<t'$. Note that by \Cref{Alg:resolve}, the fact that $\eval^{\+O}(v_{i(t)})$ only make recursive calls to $\resolve(\pred_t(u))$ for some $u\in N(v_{i(t)})$ and \eqref{eq:definition-pred} we also have $t^*>t'-n$. We then have two cases:
\begin{enumerate}
    \item $t'\leq t-n$, this contradicts the maximality assumption for $t^*$.
    \item Otherwise $t'>t-n$. By $t^*\leq t-n$ and $t'<t^*+n$ we have $t'\in [t-n+1,t]$. 
      Also, by the assumption that $\+{B}_t$ happens, we have $\+E_{t'}$ happens; therefore, $\resolve(t')$ would have directly terminated without incurring any recursive call. This also leads to a contradiction and thus proves the claim.
\end{enumerate}

Let $p\defeq \min\limits_{t}\Pr{\+E_t}$, then $p>0$ by \Cref{condition:immediate-termination}. Note that by \Cref{condition:immediate-termination}, for any $t\leq t_0$, we have 
\[
\Pr{\+{B}_t}=\prod\limits_{t'=t-n+1}^{t}\Pr{\+E_{t'}}\geq (1-p)^n>0,
\]
where the first equality is by $\+E_{t}$ only depends on the randomness of procedure $\eval$, therefore all $\+E_{t'}$ are independent.

For any $L>0$, let $\+E_L$ be the event that there is a recursive call to $\resolve(t^*)$ where $t^* \le t_0-L n$.
By the claim above,
\begin{align*}
  \Pr{\+E_L} \le \Pr{\bigwedge\limits_{j=0}^{L-1}\tp{\neg \+{B}_{t_0-jn}}} = \prod\limits_{j=0}^{L-1}\Pr{\neg \+{B}_{t_0-jn}}\leq (1-p)^L,
\end{align*}
where the equality is again due to independence of $(\+E_t)_{t\leq t_0}$.
Consequently, with probability $1$, there is only a finite number of recursive calls, meaning that $\lsample(\Lambda)$ terminates with probability $1$. This establishes \Cref{item:lsample-correctness-convergence-1}.

For any $t\le 0$, since $\resolve(t)$ terminates with probability $1$,
its output distribution is well-defined. Therefore, the output distribution of $\lsample(\Lambda)$ is well-defined. For any $\eps>0$, we choose a sufficiently large $L$ such that $(1-p)^L\le \eps$.
For any $T\ge Ln-t_0$, we couple $\lsample(\Lambda)$ with $\lsample_T(\Lambda)$ by pre-sampling all random variables used in $\eval^{\+O}(v_{i(t)})$ within $\resolve(t)$ for each $t\leq 0$. Here by \Cref{condition:local-correctness}, the coupling fails if and only if $\resolve(t')$ is recursively called within $\lsample(\Lambda)$ for some $t'\leq -T$, that is, $\+E_{L}$ happens. 
By the coupling lemma, we have
\[
\DTV{\lsample(\Lambda)}{\lsample_T(\Lambda)}\leq \Pr{\+E_{L}}\leq (1-p)^L\le \eps,
\]
which proves \Cref{item:lsample-correctness-convergence-2} as we take $T\to \infty$.
\end{proof}

For any finite $T>0$ and $-T\leq t\leq 0$, we let $X_{T,t}$ be the state of $X_t$ in $\+P(T)$. The following lemma shows $\lsample_T$ indeed simulates $\+P(T)$.

\begin{lemma}\label{lemma:lsample-correctness-distribution}
  Assume that Conditions \ref{condition:local-correctness} and \ref{condition:immediate-termination} hold for $\eval^{\+O}(v)$. Then for any $\Lambda\subseteq V$, the value returned by $\lsample_T(\Lambda)$ is identically distributed as $X_{T,0}(\Lambda)$.
\end{lemma}

\begin{proof}
We maximally couple the value returned by each $\resolve(t)$ and $X_t(v_{i(t)})$ in $\+P(T)$ for each $t\leq T$ 
and claim that in this case, the value returned by each $\resolve(t)$ is exactly the same as $X_t(v_{i(t)})$ in $\+P(T)$; hence the lemma holds by the definition of $\lsample_T$ and \eqref{eq:definition-pred}.

We prove the claim by induction from time $-T$ to $0$. For each $-T\leq t<0$, let $v=v_{i(t)}$ and consider the value returned by $\resolve(\pred_t(u))$ for each $u\in N(v)$:
\begin{itemize}
    \item If $\pred_t(u)<-T$, then by \eqref{eq:definition-pred}, the value of $u$ is not updated up to time $t$ in $\+P(T)$, hence $X_t(u)=X_{-T}(u)$ and the value returned by $\resolve(\pred_t(u))$ is $X_{-T}(u)=X_t(u)$ by the initialization of $M$ in \Cref{Alg:lsample-finite}.
    \item Otherwise, $-T\leq \pred_t(u)<t$ by \eqref{eq:definition-pred} and $u\in N(v)$, so the value returned by $\resolve(\pred_t(u))$ is $X_{\pred_t(u)}=X_t(u)$ by the induction hypothesis. Hence by \Cref{Line:resolve-sample} of \Cref{Alg:resolve} and \Cref{condition:local-correctness}, both the distribution of $\resolve(t)$ and $X_t(v_{i(t)})$ is $\mu^{X_{t-1}(N(v))}(v)$ and hence can be perfectly coupled.
\end{itemize}
Hence, the claim holds, and the lemma is proved.
\end{proof}

We are now ready to prove \Cref{lemma:lsample-correctness}.

\begin{proof}[Proof of \Cref{lemma:lsample-correctness}]
By \Cref{item:lsample-correctness-convergence-1} of \Cref{lemma:lsample-correctness-convergence}, we have $\lsample(\Lambda)$ terminates with probability $1$ and its output distribution is well-defined. It remains to prove that the output distribution of $\lsample(\Lambda)$ is exactly $\mu_{\Lambda}$.

By \Cref{lemma:lsample-correctness-irreducibility}, 
we have that $\+P(T)$ is irreducible. Then, \Cref{thm-convergence} implies that
\begin{equation}\label{eq:dtv-convergence-irreducibility}
\lim_{T\to \infty}\DTV{\mu_{\Lambda}}{X_{T,0}(\Lambda)}=0,
\end{equation}
For any $T\geq 0$, by the triangle inequality, we have
\begin{align}\label{eq:triangle-inequality}
   \DTV{\mu_{\Lambda}}{\lsample(\Lambda)}
   \leq  &\DTV{\mu_{\Lambda}}{\lsample_T(\Lambda)}\\
   &+\DTV{\lsample(\Lambda)}{\lsample_T(\Lambda)}.\notag
\end{align}
%
Altogether, the theorem follows from
    \begin{align*}
  &&&\DTV{\mu_{\Lambda}}{\lsample(\Lambda)} \\
  \text{(by \eqref{eq:triangle-inequality})}
  && \le~&~\limsup_{T\rightarrow\infty}\DTV{\mu_{\Lambda}}{\lsample_T(\Lambda)} \\
  &&~&~+ \limsup_{T\rightarrow\infty}\DTV{\lsample_T(\Lambda)}{\lsample(\Lambda)}  \\
 \text{(by \Cref{lemma:lsample-correctness-convergence})} 
 &&=~&~\limsup_{T\rightarrow\infty}\DTV{\mu_{\Lambda}}{\lsample_T(\Lambda)}  \\
 \text{(by \eqref{eq:triangle-inequality})} 
 &&\le~&~\limsup_{T\rightarrow\infty}\DTV{\mu_{\Lambda}}{X_{T,0}(\Lambda)} + \limsup_{T\rightarrow\infty}\DTV{X_{T,0}(\Lambda)}{\lsample_T(\Lambda)}   \\
  \text{(by \Cref{lemma:lsample-correctness-distribution})} 
  && =~&~\limsup_{T\rightarrow\infty}\DTV{\mu_{\Lambda}}{X_{T,0}(\Lambda)}   \\
  \text{(by \eqref{eq:dtv-convergence-irreducibility})} 
  && =~&~0.  \qedhere
\end{align*}
\end{proof}

\section{Application: a local sampler for spin systems with soft constraints}\label{sec:local}

In this section, we construct a new marginal sampling oracle for $q$-spin systems with soft constraints. We will use this oracle to build our local sampler and prove Theorems \ref{theorem:local-sampler} and \ref{theorem:inference-permissive}. Our construction is inspired by a simple rejection sampling procedure for sampling from $\mu^{\sigma}_v$, given the neighborhood configuration $\sigma \in [q]^{N(v)}$ for some $v \in V$. 

This rejection sampling procedure is described as follows:
    \begin{enumerate}
        \item Propose a random value $c\in [q]$ distributed according to $\lambda_v$;\label{item:propose}
        \item With probability $\prod\limits_{e=(u,v)\in E}A_e(\sigma(u),c)$, accept the proposal and return $c$ as the final outcome; 
        Otherwise, reject the proposal and go to Step~(\ref{item:propose}). 
    \end{enumerate}


Note that the well-definedness of the above procedure follows from \Cref{cond:main}, which ensures that all $\lambda_v$ and $A_e$ are normalized. Given a $q$-spin system $\+S = (G = (V, E), \bm{\lambda}, \bm{A})$, for any vertex $v \in V$, we can then define a marginal sampling oracle at $v$ based on the rejection sampling procedure.

\begin{algorithm}
\caption{a marginal sampling oracle for $q$-spin systems with soft constraints} \label{Alg:coupler-spin-system}
\SetKwInput{KwData}{Oracle access}
\KwIn{A $q$-spin system $\+S=(G=(V,E),\bm{\lambda},\bm{A})$, a vertex $v\in V$.}
\KwOut{A value $X \in [q]$.}
\KwData{$\+O(u)$ for each $u\in N(v)$.}
Sample an infinite long sequence of i.i.d.~tuples $\{(c_i,(r_{i,u})_{u\in N(v)})\}_{1\leq i<\infty}$ where
    each $c_i\in [q]$ is distributed as $\lambda_v$ and each $r_{i,u}$ is chosen uniformly from $[0,1]$\;
\label{line:interpret}
$i^* \gets \min \{i \mid \forall e=(u,v)\in E, r_{i,u}<A_e(\+O(u), c_i)\}$\;
\label{line:coupler-neighbors}
\Return $c_{i^*}$\;
\end{algorithm}

We remark that in \Cref{line:coupler-neighbors} of \Cref{Alg:coupler-spin-system}, such an $i^*$ always exists because $A_e$ satisfies \Cref{cond:main}. We now present the following lemma.

\begin{lemma}\label{lemma:coupler-correctness}
Suppose that the input $q$-spin system $\+S$ satisfies \Cref{cond:main}. Then, \Cref{Alg:coupler-spin-system} implements a marginal sampling oracle at $v$.
\end{lemma}

\begin{proof}
Given a $q$-spin system $\+S=(G=(V,E),\bm{\lambda},\bm{A})$, in \Cref{Alg:coupler-spin-system}, for each $i\geq 1$ and each $u\in N(v)$, recall that each $c_i$ is chosen distributed as $p\in \Delta_q$ where $p(x) \propto \lambda_v(x)$ and each $r_{i,u}$ is independently chosen uniformly from $[0,1]$. Let $\+{D}_i$ be the event that 
\[
\+D_i: \forall e=(u,v)\in E, r_{i,u}<A_e(\+O(u), c_i),
\]
then for any $x\in [q]$, note that $\+O(u_j) = \sigma(u_j)$ under assumption, we have
\begin{align*}
    \Pr{c_i=x \mid \+{D}_i} 
    &= \frac{\Pr{c_i=x \wedge \+{D}_i}}{\Pr{\+{D}_i}} \\
    &= \frac{\lambda_v(x)\prod\limits_{e=(u,v)\in E}A_e(\sigma(u),x)}{\sum\limits_{c\in [q]}\left(\lambda_v(c)\prod\limits_{e=(u,v)\in E}A_e(\sigma(u),c)\right)} = \mu_v^\sigma(x).
\end{align*}
Let $i^*$ be the smallest index chosen in \Cref{line:coupler-neighbors} of \Cref{Alg:coupler-spin-system}, i.e., 
$i^* = \min \{i \mid \+{D}_i\}$.
the output of \Cref{Alg:coupler-spin-system} follows the distribution of $c_{i^*}$ conditioning on $\+{D}_{i^*}$, concluding the proof of the lemma.
\end{proof}

The marginal sampling oracle in \Cref{Alg:coupler-spin-system} as originally designed would require a significant number of oracle calls in \Cref{line:coupler-neighbors}, potentially violating the efficiency condition outlined in \Cref{condition:fast-termination}. 
The key optimization is to invoke the oracle $\+O(u)$ only when necessary for each neighbor $u \in N(v)$, rather than for every iteration in \Cref{Alg:coupler-spin-system}.
Formally, assuming \Cref{cond:main} holds, in \Cref{line:coupler-neighbors} of \Cref{Alg:coupler-spin-system}, if $r_{i,u} < \lb$, where $\lb = C(\Delta, \delta) \defeq 1 - \frac{1-\delta}{2\Delta}$, then the inequality $r_{i,u}<A_e(\+O(u),c_i)$ will hold true regardless of the value of $\+O(u)$. This is because the term $\lb$ is chosen such that $r_{i,u}$ is sufficiently small to ensure success in the comparison without needing the actual value of $\+O(u)$. Consequently, it becomes unnecessary to call $\+O(u)$ when $r_{i,u} < \lb$.
With this idea of optimization, we propose the following implementation of $\eval^{\+O}(v)$, presented in \Cref{Alg:eval}, which builds upon the above idea to efficiently sample without violating the fast termination condition.




\begin{algorithm}
\caption{$\eval^{\+O}(v)$ } \label{Alg:eval}
\SetKwInput{KwData}{Global variables}
\KwIn{A $q$-spin systems $\+S=(G=(V,E),\bm{\lambda},\bm{A})$, a vertex $v\in V$.}
\KwOut{A value $c \in [q]$.}
 Sample an infinite long sequence of i.i.d.~tuples $\{(c_i,(r_{i,u})_{u\in N(v)})\}_{1\leq i<\infty}$, where
    each $c_i\in [q]$ is distributed as $\lambda_v$ and each $r_{i,u}$ is chosen uniformly from $[0,1]$\;\label{Line:eval-sample}
\For{$i=1,2,...$\label{Line:reject}}{
$\textit{flag} \gets 1$\;
\For{$e=(u,v)\in E$\label{Line:eval-for}}{
\If{$r_{i,u}\geq \lb$\label{Line:eval-cond}}{
\lIf{$r_{i,u}\geq A_{e}(\+O(u),c_i)$\label{Line:eval-if}}{$\textit{flag} \gets 0$}
}
}
\lIf{$\textit{flag}=1$\label{Line:eval-return}}{\Return $c_i$}
}
\end{algorithm}

\begin{remark}[principle of deferred decision]\label{remark:lazy-samples}
    In \Cref{Line:eval-sample} of \Cref{Alg:eval}, we are required to sample an infinitely long sequence $\{(c_i,\{r_{i,u}\}_{u\in N(v)})\}_{1\leq i<\infty}$. 
    Obviously, it is not feasible to directly sample an infinite number of random variables for implementation.  
    Instead, we adopt the principle of deferred decision: each $c_i$ and $r_{i,u}$ is generated only when they are accessed in  \Cref{Line:eval-for,Line:eval-if} of \Cref{Alg:eval}. 
\end{remark}

Next, we show that the marginal sampling oracle $\eval^{\+O}(v)$ in \Cref{Alg:eval} satisfies the conditions for both correctness and efficiency. 

\begin{lemma}\label{lemma:eval-correctness}
Suppose that the input $q$-spin system $\+S$ satisfies \Cref{cond:main}. Then, the marginal sampling oracle $\eval^{\+O}(v)$ implemented as in \Cref{Alg:eval} satisfies both \Cref{condition:local-correctness} and \Cref{condition:immediate-termination}.
\end{lemma}

\begin{proof}

Condition \ref{condition:local-correctness} can be verified directly by \Cref{lemma:coupler-correctness} and comparing Algorithms \ref{Alg:coupler-spin-system} and \ref{Alg:eval}.

For Condition \ref{condition:immediate-termination}, recall that $\+E_v$ is the event that $\eval^{\+O}(v)$ in \Cref{Alg:eval} terminates without any calls to $\+O$. Note that $\+E_v$ occurs if and only if \Cref{Alg:eval} terminates within $1$ round of the loop at \Cref{Line:reject} and $r_{1,u}<\lb$ holds for each $u\in N(v)$. Since $c_1$ and each $r_{1,u}$ are independent, assuming the $q$-spin system $\+S$ satisfies \Cref{cond:main} with constant $\delta>0$, we immediately have
\begin{align*}
\Pr{\+E_t} \geq C^{\Delta}>0. 
\end{align*}
It verifies that $\eval^{\+O}(v)$ satisfies Condition \ref{condition:immediate-termination} assuming Condition \ref{cond:main} holds.
\end{proof}

\begin{lemma}\label{lemma:eval-efficiency}
Suppose that the input $q$-spin system $\+S$ satisfies \Cref{cond:main}. Then, the marginal sampling oracle $\eval^{\+O}(v)$ implemented as in \Cref{Alg:eval} satisfies \Cref{condition:fast-termination}.
\end{lemma}

\begin{proof}
Fix some $v\in V$. Also fix some $\sigma\in [q]^{N(v)}$ and assume that $\+O(u)$ returns $\sigma(u)$ within $\eval^{\+O}(v)$. 
Within each round of the outer for loop at \Cref{Line:reject} of \Cref{Alg:eval}, 
the expected total number of oracle calls to $\+O(u)$ for any $u\in N(v)$ is given by:
\begin{align}
\E{\text{number of calls to }\+O(\cdot)\text{ in one iteration}} =\sum\limits_{c\in[q]}\left(\lambda_v(c) \sum_{e=(u,v)\in E} (1-\lb)\right).\label{eq:expect-one-round}
\end{align}

 Let $\+I^{\sigma}_v$ be the number of executions of the outer for loop at \Cref{Line:reject} in $\eval^{\+O}(v)$. Note that $\+I^{\sigma}_v$ corresponds exactly to the number of executions of \Cref{item:propose} in the rejection sampling. Therefore,
\begin{align}
\E{\+{I}^{\sigma}_v} = \frac{1}{\sum\limits_{c\in[q]}\left(\lambda_v(c)\prod\limits_{e=(u,v)\in E}A_e(\sigma(u),c)\right)}.
\label{eq:expect-number-of-round}
\end{align}


Note that the number of oracle calls in each iteration are i.i.d.~random variables with finite mean. 
Applying Wald’s equation, we have:
\begin{align*}
\E{\+{T}^{\sigma}_v}
&=
\E{\+{I}^{\sigma}_v}\cdot\E{\text{number of calls to }\+O(\cdot)\text{ in one iteration}}.
\end{align*}

It then follows from  \eqref{eq:expect-one-round} and \eqref{eq:expect-number-of-round} that:
\begin{align*}
 &&   \E{\+{T}^{\sigma}_v}
   =&\, \frac{\sum\limits_{c\in[q]}\left(\lambda_v(c) \sum\limits_{e=(u,v)\in E} (1-\lb)\right)}{\sum\limits_{c\in[q]}\left(\lambda_v(c)\prod\limits_{e=(u,v)\in E}A_e(\sigma(u),c)\right)}\\
\text{(by \Cref{cond:main})}
&& 
  \leq &\, \frac{\sum\limits_{c\in [q]}\left(\lambda_v(c)\cdot (1-\lb)\cdot |N(v)| \right)}{\sum\limits_{c\in[q]}\left(\lambda_v(c)\cdot C^{\Delta}\right)} \\
 \text{(by $\lb=1-\frac{\delta}{2\Delta}$)}        
 &&\leq &\, \frac{\sum\limits_{c\in [q]}(\lambda_v(c)(1-\delta)/2)}{\sum\limits_{c\in [q]}(\lambda_v(c)(1+\delta)/2)} \\
  &&       < &\, 1.
\end{align*}

It proves that $\eval^{\+O}(v)$ satisfies Condition \ref{condition:fast-termination} assuming Condition \ref{cond:main} holds.
\end{proof}

We are now ready to prove \Cref{theorem:local-sampler,theorem:inference-permissive}.

\begin{proof}[Proof of \Cref{theorem:local-sampler}]
    We use \Cref{Alg:lsample} as our local sampler, where the subroutine $\eval^{\+O}(v)$ is implemented by \Cref{Alg:eval} (using the principle of deferred decision as explained in \Cref{remark:lazy-samples}). 
    
    Here, the correctness of sampling follows from \Cref{lemma:lsample-correctness,lemma:eval-correctness}.
    
    For efficiency, by \Cref{lemma:lsample-efficiency,lemma:eval-efficiency}, we have that the expected number of $\resolve$  calls is $O(|\Lambda|)$. Also, note that each outer loop either terminates directly or results in at least one call to $\resolve$. Hence, the overall running time is bounded by $\Delta \log q$ times the total number of $\resolve$  calls, which is $O(|\Lambda|\Delta\log q)$ in expectation. 
\end{proof}

\begin{proof}[Proof of \Cref{theorem:inference-permissive}]
    Note that by the self-reducibility of spin systems with soft constraints (i.e., \Cref{cond:main} holds under arbitrary pinning), we have the uniform lower bound:
    \[
     \forall v\in V,c\in [q],\quad \mu^{\sigma}_v(c)\geq \frac{C(\Delta,\delta)^{\Delta}}{q}\geq \frac{1}{2q}.
    \]

    This ensures that the marginal probabilities $\mu^{\sigma}_v(x)$ are all bounded away from zero. Therefore, by the Chernoff bound, for each $x \in Q_v$, the value $\mu^{\sigma}_v(x)$ can be estimated to within a multiplicative error of $(1 \pm \varepsilon)$ with probability at least $0.9$ using $O(q/\varepsilon^2)$ approximate samples.

    According to \Cref{theorem:local-sampler}, each such sample can be generated in expected time $O(|\Lambda|\Delta\log q)$, where we easily handle the conditioning in $\sigma$ by setting $M(t)=\sigma_{(v_{i(t)})}$ whenever $\resolve(t)$ is called for some $v_{i(t)}\in \Lambda$,  so the total expected cost to estimate all $\mu^{\sigma}_v(x)$ for $x \in Q_v$ is bounded by
\begin{align}\label{eq:inference-complexity-bound}
O\left(|\Lambda| q^2\log q\Delta\varepsilon^{-2}\right).
\end{align}
Applying Markov's inequality, the probability that the total cost exceeds this bound is at most $0.1$. Therefore, by truncating the algorithm’s running time to \eqref{eq:inference-complexity-bound}, we obtain a bounded-cost algorithm which, with probability at least $0.9 - 0.1 = 0.8$, returns estimates of all $\mu^{\sigma}_v(x)$ within a multiplicative factor of $(1 \pm \varepsilon)$ for all $x \in Q_v$.

Finally, to boost the success probability to at least $1 - \delta/q$, we repeat this procedure independently $O(\log(q/\delta))$ times and take the median of the resulting estimates. This yields the desired algorithm claimed in the theorem by applying Chernoff's bound again.
\end{proof}

\section{Extension: a local sampler for \texorpdfstring{$q$}{q}-colorings}\label{sec:coloring}
In this section, we show how to extend our CTTP framework to design local samplers for a prototypical spin system with truly repulsive hard constraints: the proper $q$-colorings.

A key challenge in applying the CTTP framework to problems with hard constraints lies in the absence of unconditional marginal lower bounds. In the systematic scan Glauber dynamics $\+P(T)=(X_t)_{t=-T}^{0}$ for $q$-colorings, determining the outcome of an update at node $v=v_{i(t)}$ at time $t$ requires knowledge of the color set $S_t$, defined as the set of colors assigned to the neighbors of $v$ in the configuration $\sigma_t$:
\begin{equation}\label{eq:definition-St}
S_t \defeq \{X_t(u) \mid u \in N(v)\}.
\end{equation}
To correctly perform the update at time $t$, one needs to sample uniformly from $[q] \setminus S_t$. However, if for any neighbor $u\in N(v)$, the outcome of its most recent update occurring at time $\pred_t(u)$ is unknown, then the update at time $t$ cannot be resolved with positive probability. In other words, the immediate termination condition (\Cref{condition:immediate-termination}) does not hold without additional information.

To generate a uniform sample from the set $[q] \setminus S_t$, we could use a rejection sampling procedure similar to that described in the previous section: propose a color $c$ uniformly at random from $[q]$, and accept it if $c \notin S_t$. Since $S_t$ contains at most $\Delta$ colors, when $q > C\Delta$ for some sufficiently large constant $C$, this rejection sampling succeeds with constant probability. However, a key challenge remains: determining whether $c \in S_t$ requires knowledge of $X_t(u)$ for each neighbor $u \in N(v)$, which may lead to endless recursion.

The crucial observation is that for each neighbor $u \in N(v)$, at the time of its last update $t' = \pred_u(t)$, the probability that $X_{t'}(u) = c$ is at most $\frac{1}{|[q] \setminus S_{t'}|}$, which is bounded above by $\frac{2}{q}$ under the assumption $q \geq 2\Delta$. This allows us to implement a probabilistic filter to test whether $X_{t'}(u) \neq c$ as follows:
\begin{itemize}
\item With probability $1 - \frac{2}{q}$, we directly \emph{certify} that $X_{t'}(u) \neq c$;
\item With the remaining probability $\frac{2}{q}$, we first check whether $c \in S_{t'}$:
\begin{itemize}
\item If $c \in S_{t'}$, we can immediately \emph{certify} that $X_{t'}(u) \neq c$. This condition can be verified recursively by invoking a similar procedure for each neighbor $w \in N(u)$;
\item Otherwise, we set $X_{t'}(u) = c$ with probability $\frac{q/2}{|[q] \setminus S_{t'}|}$, and otherwise \emph{certify} that $X_{t'}(u) \neq c$.
\end{itemize}
\end{itemize}

Observe that the filtering procedure terminates immediately with probability at least $1 - \frac{2}{q}$. When $q$ is sufficiently large, this high probability of immediate termination serves as the basis of the recursion, allowing us to avoid infinite recursion. 
To fully implement this filtering process, two aspects remains to be specified: what it means to ``certify'' that a color cannot appear at a given timestamp $t$, and how to simulate a trial with an unknown success probability $\frac{q/2}{|[q] \setminus S_{t'}|}$. 
We briefly outline these below:
\begin{itemize}
    \item To certify that a certain color cannot appear at timestamp $t$, we have an auxiliary  mapping $L : \mathbb{Z} \to 2^{[q]}$ maintained globally outside the recursion, where each $L(t)$ is initially set to $[q]$. This data structure records the set of \emph{available colors} at each timestamp. The color assigned at time $t$ is interpreted as being uniformly sampled from $L(t) \setminus S_t$. Thus, to certify that $c$ cannot be the outcome of the update at time  $t$, we simply remove $c$ from $L(t)$.
    \item Although we cannot directly compute the probability $\frac{q/2}{|[q] \setminus S_{t'}|}$ since $S_{t'}$ is not explicitly known, we can query the membership of certain colors in $S_{t'}$ through recursive calls. This allows us to employ a \emph{Bernoulli factory algorithm} to simulate a Bernoulli trial with the correct success probability $\frac{q/2}{|[q] \setminus S_{t'}|}$.
\end{itemize}

\begin{remark}[adaptivity of the grand coupling]\label{remark:grand-couping-adaptivity}
As noted in \Cref{remark:implicit-grand-coupling}, simulating the Markov chain in this manner implicitly defines a grand coupling over all starting configurations. In the procedure described above, this grand coupling is \emph{adaptive}: the coupling decision at a given timestamp $t$ may depend on the outcome of the coupling at an earlier time $t' < t$. This adaptivity is a crucial feature of our local sampler for $q$-colorings. Later, we will show that it does not compromise the correctness of the Coupling Towards The Past (CTTP) framework. We also note that similar adaptive strategies have been employed in the design of several perfect sampling algorithms based on the Coupling From The Past (CFTP) method~\cite{bhandari2020improved,jain2021perfectly}. However, our strategy fundamentally differs from those in the CFTP framework, due to the additional constraints imposed by local samplers.
\end{remark}



\subsection{The local sampler}
We now formally present our local sampler for $q$-colorings. As discussed above, our local sampler is built on the CTTP framework,  with a key modification that allows partial information whether an updated outcome equals some particular color to be obtained. 
The local sampler relies on two core subroutines:
\begin{itemize}
    \item $\resolve(t)$, which takes as input a timestamp $t\leq 0$, and returns the outcome of the update at time $t$ (i.e., the color to which the vertex is updated);
    \item $\chk(t,c)$, which takes as input a timestamp $t\leq 0$ as well as a color $c\in [q]$, and determines whether the outcome of the update at time $t$ equals $c$;
\end{itemize}

Our local sampler for $q$-colorings is formally described in \Cref{Alg:lsample-coloring}.

\begin{algorithm}
\caption{$\lsample(\Lambda;M,L)$} \label{Alg:lsample-coloring}
\SetKwInput{KwData}{Global variables}
\KwIn{A $q$-coloring instance $G=(V,E)$, a subset of variables $\Lambda\subseteq V$.}
\KwOut{A random configuration $X \in [q]^{\Lambda}$.}
\KwData{Two mappings $M: \mathbb{Z}\to [q]\cup\{\perp\}, L: \mathbb{Z} \to 2^{[q]}$.}
$X\gets \emptyset, M\gets \perp^{\mathbb{Z}},L\gets [q]^{\mathbb{Z}}$\label{Line:local-sampler-coloring-initialization}\;
\ForAll{$v\in \Lambda$}{
    $X(v)\gets \resolve(\pred_0(v);M,L)$\;\label{Line:local-sampler-coloring-resolve}
}
\Return $X$\;
\end{algorithm}

In \Cref{Alg:lsample-coloring}, two global mappings, $M$ and $L$, are maintained for memoization to ensure the consistency of the algorithm. Specifically, $M(t)$ stores the resolved values of updates at each time $t$, and $L(t)$ stores the remaining set of available colors for unresolved updates at each time $t$. 

We now formally present the two core procedures, $\resolve(t)$ and $\chk(t,c)$, which underpin our construction of the local sampler. Similar to the case of spin systems with soft constraints, these procedures are recursively defined, making calls to themselves on earlier timestamps. To facilitate their definition, we adopt a variant of the abstraction for local sampling procedures, resembling the $\eval^{\+O}(v)$ introduced in \Cref{definition:locally-defined-grand-coupling}, extended to allow access to partial information.

\begin{definition}[local sampling procedures with access to partial information]\label{definition:partial-information}
For each vertex $v\in V$ and color $c\in [q]$, we define the procedures 
$\eval^{\+O}(v)$ and $\eval 
^{\+O}(v,c)$, which make oracle queries to:
\begin{itemize}
\item $\+O(u)$, which consistently returns a value $c_u\in [q]$ for each neighbor $u\in N(v)$;
\item $\+O(u,c)$, which consistently returns a binary value $x_{u,c}\in \{0,1\}$ for each $u\in N(v)$ and $c\in [q]$, such that $x_{u,c}=1$ if and only if $c_u=c$.
\end{itemize}
In addition, since local sampling steps require knowledge of the set of currently available colors, finite slices of the global mapping $L$ may be passed by reference to $\eval^{\+O}(v)$ and $\eval ^{\+O}(v,c)$.
\end{definition}

We remark that \Cref{definition:partial-information} differs from the marginal sampling oracles in \Cref{definition:locally-defined-grand-coupling} as it imposes no requirement on the distribution of the outcomes produced by $\eval^{\+O}(v)$ and $\eval 
^{\+O}(v,c)$.

The procedures $\resolve(t)$ and $\chk(t,c)$ are detailed in \Cref{Alg:resolve-coloring,Alg:check}, respectively. 

The procedure $\resolve(t)$ first checks whether the update at time $t$ has already been resolved,  and if not, invokes $\eval^{\+O}(v)$ to resolve it, passing the list of available colors $L(t)$ by reference. Similarly, the procedure $\chk(t,c)$ performs a series of preliminary checks before invoking $\eval^{\+O}(v,c)$ determine equality to $c$, also passing $L(t)$ by reference and potentially updating $L(t)$ at the end. 

An important detail occurs in \Cref{Line:check-list-size,Line:check-resolve-instead} of $\chk(t,c)$: if the size of the list $L(t)$ falls below $50\Delta$, the algorithm invokes $\resolve(t)$ to fully resolve the outcome at time $t$ rather than merely checking equality to $c$. This mechanism ensures that the list size $|L(t)|$ remains above $50\Delta$ whenever it is passed to either $\eval^{\+O}(v)$ or $\eval^{\+O}(v,c)$ during the algorithm's execution, assuming $q \geq 50\Delta$.

\begin{algorithm}
\caption{$\resolve(t; M,L)$} \label{Alg:resolve-coloring}
\SetKwInput{KwData}{Global variables}
\KwIn{A $q$-coloring instance $G=(V,E)$, a timestamp $t\in \mathbb{Z}$.}
\KwOut{A random configuration $X \in [q]^{\Lambda}$.}
\KwData{Two mappings $M: \mathbb{Z}\to [q]\cup\{\perp\}, L: \mathbb{Z} \to 2^{[q]}$.}
\lIf(\tcp*[f]{check if the outcome is already resolved}){$M(t) \neq \perp$}{\Return $M(t)$\label{Line:resolve-coloring-finite-memoization}}
$M(t)\gets \eval^\+O(v_{i(t)};L(t))$, with{
\begin{itemize}
\item $\+O(u)$ replaced by $\resolve(\pred_t(u);M,L)$ for each $u\in N(v_{i(t)})$\;\label{Line:resolve-coloring-sample}
\item $\+O(u,c)$ replaced by $\chk(\pred_t(u),c;M,L)$ for each $u\in N(v_{i(t)})$\;
\end{itemize}}
\Return $M(t)$\; \label{Line:resolve-coloring-final-return}
\end{algorithm}

\begin{algorithm}
\caption{$\chk(t,c;M,L)$ } \label{Alg:check}
\SetKwInput{KwData}{Global variables}
\KwIn{A $q$-coloring instance $G=(V,E)$, a timestamp $t\in \mathbb{Z}$, a color $c\in [q]$.}
\KwOut{A binary value $x\in \{0,1\}$.}
\KwData{Two mappings $M: \mathbb{Z}\to [q]\cup\{\perp\}, L: \mathbb{Z} \to 2^{[q]}$.}

\If(\tcp*[f]{check if the outcome is already resolved}){$M(t)\neq \perp$}{
\If{$M(t)=c$} {\Return $1$\;} \Else {\Return $0$\;}
}
\If(\tcp*[f]{check if $c$ is already not available}){$c\notin L(t)$}{
    \Return $0$; 
}
\If(\tcp*[f]{resolve the full outcome instead when $\abs{L(t)}\leq 50\Delta$}){$\abs{L(t)}\leq 50\Delta$\label{Line:check-list-size}}{
    $\resolve(t;M,L)$; \label{Line:check-resolve-instead}
}
$x\gets \eval^\+O(v_{i(t)},c;L(t))$, with{
\begin{itemize}
\item $\+O(u)$ replaced by $\resolve(\pred_t(u);M,L)$ for each $u\in N(v_{i(t)})$\;\label{Line:check-sample}
\item $\+O(u,c)$ replaced by $\chk(\pred_t(u),c;M,L)$ for each $u\in N(v_{i(t)})$\;
\end{itemize}}
\If{$x=1$}{
    $M(t)\gets c$\;
}
\Else{
    $L(t)\gets L(t)\setminus \{c\}$\label{Line:check-list-remove}\;
}
\Return $x$\;
\end{algorithm}

We now present our specific construction of local sampling procedures for $q$-colorings. To fully evaluate a specific outcome of the update and realize $\eval^{\+O}(v;L)$, we employ the strategy  described earlier. Let $S=\{\+O(u)\mid u\in N(v)\}$. To obtain a uniform sample from $L\setminus S$, we repeatedly select a  color $c\in L$ uniformly at random and check whether $c\in S$ by querying the oracle $\+O(u,c)$ for each $u\in N(v)$. This procedure is detailed in \Cref{Alg:evaluate-coloring}.

\begin{algorithm}
\caption{$\eval^{\+O}(v;L)$} \label{Alg:evaluate-coloring}
\SetKwInput{KwData}{Oracle access}
\SetKwInput{KwMaintain}{Global variables}
\KwIn{A $q$-coloring instance $G=(V,E)$, a variable $v\in V$, and a list of available colors $L\subseteq [q]$, with the promise that $|L|\geq 50\Delta$.}
\KwOut{A value $c \in [q]$.}
\KwData{$\+O(u)$ for each $u\in N(v)$ and $\+O(u,c)$ for each $u\in N(v)$ and $c\in [q]$.}
\KwMaintain{a list of available colors $L\subseteq [q]$, with the promise that $c\notin L$ and $|L|\geq 50\Delta$.}
\While{$\True$}{
choose $c\in L$ uniformly at random\;\label{Line:evaluate-coloring-choose}
\If{$\+O(u,c)=0$ for each $u\in N(v)$ \label{Line:evaluate-coloring-check} }{\Return $c$\;}
$L\gets L\setminus \{c\}$\;\label{Line:evaluate-coloring-update}
}
\end{algorithm}

To determine whether the outcome equals a specific color $c$ and realize $\eval^{\+O}(v, c; L)$, we leverage the observation that when the size of the available color list satisfies $|L| \geq 2\Delta \geq 2|S|$, the probability that a uniform sample from $L \setminus S$ equals $c$ is small. This motivates the following strategy to simulate a coin that always returns $0$ when $c \in S$, and when $c \notin S$, returns $1$ with probability $1/|L \setminus S|$ and $0$ otherwise:
\begin{itemize}
    \item With probability $1-\frac{1}{|L|/2}$, return $0$ directly;
    \item Otherwise:
    \begin{itemize}
        \item If $c\in S$, return $0$;
        \item Otherwise, return $0$ with probability $\frac{|L|/2}{|L\setminus S|}$, and return $1$ otherwise.
    \end{itemize}
\end{itemize}
When the first step does not immediately return $0$, we query $\+O(u, c)$ for each $u\in N(v)$ to check whether $c \in S$. If $c\notin S$, a challenge arises: the set  $S$ is unknown to the algorithm, so we cannot directly simulate a coin with success probability $\frac{|L|/2}{|L \setminus S|}$. While one could retrieve $S$ by querying $\+O(u)$ for every neighbor $u \in N(v)$,  this incurs a prohibitive number of recursive calls. 
Fortunately, we can efficiently simulate a coin with success probability $\frac{|L \setminus S|}{|L|}$ using the following procedure:
\begin{itemize}
    \item Choose a color $c\in L$ uniformly at random.
    \item Determine whether $c\in S$: query $\+O(u,c)$ for all $u\in N(v)$; if any query returns 1, output 0; otherwise, output 1. 
\end{itemize}
Moreover, given access to such a $\frac{|L\setminus S|}{|L|}$-coin, we can simulate a coin with success probability $\frac{|L|/2}{|L\setminus S|}$ by solving an instance of the Bernoulli factory problem~\cite{vonNeumann1951}. 
Specifically, we employ a Bernoulli factory algorithm for division: given access to a coin with success probability $p=\frac{\abs{L\setminus S}}{\abs{L}}$, we simulate a coin with success probability $\frac{1}{2p}$. 

Our complete procedure for checking whether the outcome equals a specific color $c$ is described in \Cref{Alg:evaluate-coloring-partial}, while the subroutine for simulating a $\frac{\abs{L\setminus S}}{\abs{L}}$-probability coin is described in \Cref{Alg:access-p-coin}. 
 
\begin{algorithm}
\caption{$\eval^{\+O}(v,c;L)$} \label{Alg:evaluate-coloring-partial}
\SetKwInput{KwData}{Oracle access}
\SetKwInput{KwMaintain}{Global variables}
\KwIn{A $q$-coloring instance $G=(V,E)$, a variable $v\in V$, a color $c\in [q]$ }
\KwOut{A binary value $x \in \{0,1\}$.}
\KwData{$\+O(u)$ for each $u\in N(v)$ and $\+O(u,c)$ for each $u\in N(v)$ and $c\in [q]$.}
\KwMaintain{a list of available colors $L\subseteq [q]$, with the promise that $c\notin L$ and $|L|\geq 50\Delta$.}
with probability $1-\frac{1}{|L|/2}$ \Return $0$\label{Line:evaluate-coloring-partial-return-uniformity}\;
\ForAll{$u\in N(v)$}{
    If $\+O(u,c)=1$ \Return $0$;
}
let $\+C$ be an $\frac{\abs{L\setminus S}}{\abs{L}}$-probability coin where $S=\{\+O(u)\mid u\in N(v)\}$, generated as in \Cref{Alg:access-p-coin}\;
with probability $\frac{|L|/2}{\abs{L\setminus S}}$ return $1$, otherwise return $0$, where the probability $\frac{|L|/2}{\abs{L\setminus S}}$ is generated using the Bernoulli factory algorithm for division~\cite{morina2021extending}, given access to $\+C$. \label{Line:evaluate-coloring-bernoulli-factory}
\end{algorithm}

\begin{algorithm}
\caption{access to an ${|L\setminus S|}/{|L|}$-probability coin} \label{Alg:access-p-coin}
\SetKwInput{KwData}{Oracle access}
\SetKwInput{KwMaintain}{Global variables}
\KwIn{A $q$-coloring instance $G=(V,E)$, a variable $v\in V$, a color $c\in [q]$}
\KwOut{A binary value $x \in \{0,1\}$.}
\KwData{$\+O(u)$ for each $u\in N(v)$ and $\+O(u,c)$ for each $u\in N(v)$ and $c\in [q]$.}
\KwMaintain{A list of available colors $L\subseteq [q]$.}
choose $c_0\in L$ uniformly at random\;\label{Line:access-p-coin-sample}
\ForAll{$u\in N(v)$}{
    \lIf{$\+O(u,c_0)=1$}{\Return $0$}
}
\Return $1$\;
\end{algorithm}

\begin{remark}[implementation of the algorithm]\label{remark:algorithm-execution}
Note that our local sampler for $q$-colorings (Algorithms \ref{Alg:lsample-coloring}–\ref{Alg:access-p-coin}) involves a considerable number of recursive calls. To implement this algorithm efficiently, all recursive calls to 
$\resolve(t)$ and $\chk(t,c)$ are managed using a stack. At each step, the algorithm pops the call at the top of the stack --- either $\resolve(t)$ or $\chk(t,c)$ --- executes it, updates the values of $L(t)$ and $M(t)$ in the global data structure if needed, and pushes any new recursive calls onto the stack. 
Since each call to $\resolve(t)$ or $\chk(t,c)$ only recurses on strictly smaller timestamps $t'<t$, the recursions will be executed correctly as above.
\end{remark}

The Bernoulli factory for division used in \Cref{Alg:evaluate-coloring-partial} is achieved by a combination of existing constructions~\cite{Nacu2005FastSO,Hub16Bernoulli,morina2021extending}. Here, we present its correctness and efficiency. 
\begin{lemma}[correctness and efficiency of the Bernoulli factory algorithm]\label{lemma:correctness-Bernoulli-factory}
There is a Bernoulli factory algorithm accessing a $\xi$-coin with the promise that $\frac{1}{2}+\zeta\leq \xi\leq 1$ for some $\zeta>0$, terminates with probability $1$, returns $1$ with probability $\frac{1}{2\xi}$ and returns $0$ otherwise.  Moreover, the expected number of calls to the $\xi$-coin is at most $9.5\xi^{-1}\zeta^{-1}$.
\end{lemma}

The proof of \Cref{lemma:correctness-Bernoulli-factory} will be presented in \Cref{sec:bernoulli-factory}, where we also present the explicit construction of the Bernoulli factory algorithm for division.

\subsection{Correctness of the local sampler}

We now proceed to prove the correctness of the local sampler, stated in \Cref{lemma:partial-rejection-conditional-correctness}. Here, we establish only the \emph{conditional} correctness under the assumption that the local sampler always terminates. The proof of termination, along with the efficiency analysis of the algorithm, is presented in the next subsection.

\begin{lemma}[conditional correctness of the local sampler for $q$-colorings]
\label{lemma:partial-rejection-conditional-correctness}
    Suppose that the input $q$-coloring instance $G=(V,E)$ satisfies $q\geq 65\Delta$ where $\Delta\geq 1$ is the maximum degree of $G$. For any $\Lambda\subseteq V$, suppose that $\lsample(\Lambda)$ terminates almost surely, then the output $X$ of $\lsample(\Lambda)$ follows the law $\mu_{\Lambda}$, where $\mu$ denotes the uniform distribution over all proper $q$-colorings of $G$.
\end{lemma}

To prove \Cref{lemma:partial-rejection-conditional-correctness}, we first establish the correctness of the local procedures  $\eval^{\+O}(v;L)$ and $\eval^{\+O}(v,c;L)$. 
\begin{lemma}[local correctness of $\eval^{\+O}(v;L)$ and $\eval^{\+O}(v,c;L)$]\label{lemma:coloring-local-correctness}
Given as input a $q$-coloring instance $G=(V,E)$. For any vertex $v\in V$, color $c\in [q]$ and a list of available colors $L\subseteq [q]$ satisfying $\abs{L}\geq 50\Delta$ where $\Delta\geq 1$ is the maximum degree of $G$, assuming $\+O(u)$ for each $u\in N(v)$ and $\+O(u,c)$ for each $u\in N(v)$ and $c\in [q]$ as specified in \Cref{definition:partial-information} and let $S=\{\+O(u)\mid u\in N(v)\}$:
\begin{enumerate}
    \item $\eval^{\+O}(v;L)$ returns a uniform sample from $L\setminus S$;\label{item:coloring-local-correctness-1}
    \item $\eval^{O}(v,c;L)$:
    \begin{itemize}
        \item returns $0$ if $c\in S$;
        \item returns $1$ with probability $\frac{1}{\abs{L\setminus S}}$ and $0$ with probability $1-\frac{1}{\abs{L\setminus S}}$ otherwise.
    \end{itemize}\label{item:coloring-local-correctness-2}
\end{enumerate}
\end{lemma}
\begin{proof}
For $\eval^{\+O}(v;L)$, by $\abs{L}\geq 50\Delta$, Lines \ref{Line:evaluate-coloring-choose}-\ref{Line:evaluate-coloring-update} of \Cref{Alg:evaluate-coloring}, and \Cref{definition:partial-information}, $\eval^{\+O}(v;L)$ always returns a value from $L\setminus S$. Then \Cref{item:coloring-local-correctness-1} follows from the symmetry of all colors in $L\setminus S$.

It is direct that \Cref{Alg:access-p-coin} returns $0$ when $c\in S$. Also, according to that $\abs{L}\geq 50\Delta$ and the correctness guarantee of the Bernoulli factory algorithm (\Cref{lemma:correctness-Bernoulli-factory}), \Cref{Alg:access-p-coin}  indeed returns $1$ with probability $\frac{1}{\abs{L\setminus S}}$ and $0$ with the remaining probability when $c\notin S$. \Cref{item:coloring-local-correctness-2} then directly follows.
\end{proof}

Similar to the proof of correctness for our local sampler for spin systems with soft constraints, we introduce a finite-time variant of \Cref{Alg:lsample-coloring} for any finite $T > 0$, denoted as \Cref{Alg:lsample-finite-coloring}. This algorithm locally reconstructs the final state $X_0$ of $\+P(T)$. Still, the only difference between Algorithms \ref{Alg:lsample-coloring} and \ref{Alg:lsample-finite-coloring} lies in the initialization of the mapping $M$.

\begin{algorithm}[h]
\caption{$\lsample_T(\Lambda;M,L)$ (for $q$-colorings)} \label{Alg:lsample-finite-coloring}
\SetKwInput{KwData}{Global variables}
\KwIn{a $q$-coloring instance $G=(V,E)$, a subset of variables $\Lambda\subseteq V$}
\KwOut{A random configuration $X \in [q]^{\Lambda}$}
\KwData{Two mappings $M: \mathbb{Z}\to [q]\cup\{\perp\}, L: \mathbb{Z} \to 2^{[q]}$.}
$X\gets \emptyset, M(t)\gets X_{-T}(v_{i(t)}) \text{ for each } t\leq -T,  M(t)\gets \perp \text{ for each } t> -T$, $L\gets [q]^{\mathbb{Z}}$\label{Line:lsample-finite-coloring-initialization}\;
\ForAll{$v\in \Lambda$}{
    $X(v)\gets \resolve(\pred_0(v);M,L)$\;
}
\Return $X$\;
\end{algorithm}
We then present the crucial lemma for the correctness of the finite-time version of the algorithm. For any finite $T>0$ and $-T\leq t\leq 0$, we let $X_{T,t}$ be the state of $X_t$ in $\+P(T)$. The following lemma shows $\lsample_T$ indeed simulates $\+P(T)$ for $q$-colorings.
 \begin{lemma}\label{lemma:lsample-coloring-identical}
  Suppose that the input $q$-coloring instance $G=(V,E)$ satisfies $q\geq 50\Delta$ where $\Delta\geq 1$ is the maximum degree of $G$. For any $\Lambda\subseteq V$, the value returned by $\lsample_T(\Lambda)$ is identically distributed as $X_{T,0}(\Lambda)$.
 \end{lemma}
 \begin{proof}
 We claim that after the initialization step (Line~\ref{Line:lsample-finite-coloring-initialization}) of \Cref{Alg:lsample-finite-coloring}, it is possible to couple the randomness of the algorithm and the process $(X_t)_{-T \leq t \leq 0}$ in such a way that the following two invariants hold for any $T \geq 0$:
 \begin{enumerate}
 \item Whenever $\chk(t, c)$ is called for some timestamp $-T \leq t \leq 0$ and some color $c \in [q]$, the outcome is $1$ if and only if $X_t(v_{i(t)}) = c$, where we define $X_t = X_{-T}$ for all $t \leq -T$; \label{item:invariant-1}
\item Whenever $\resolve(t)$ is called for some timestamp $-T \leq t \leq 0$, its output is $X_t(v_{i(t)})$, where again $X_t = X_{-T}$ for all $t \leq -T$. \label{item:invariant-2}
\end{enumerate}
Once such a coupling is established, the lemma immediately follows.

We then verify the existence of this coupling by induction on the length of the process, $T$.

The base case is when $T=0$, and the invariants hold by construction due to the initialization of the mapping $M$ in \Cref{Alg:lsample-coloring}.

Assume the claim holds for some $T - 1 \geq 0$. Consider the case for $T$. We rely on two facts:
\begin{itemize}
\item The process $(X_t)_{-T \leq t \leq 0}$ forms a Markov chain.
\item Any recursive call to $\resolve(t)$ or $\chk(t, c)$ only involves timestamps $t' < t$.
\end{itemize}
These ensure that we can apply the induction hypothesis to couple the randomness for all timestamps strictly less than $t = 0$.

We now extend the coupling to the timestamp $t = 0$. Let $v = v_{i(0)}$, and
  let $S=\bigcup\limits_{u\in N(v)}X_{0}(v_{i(\pred_u(0)})$, i.e., $S$ is the set of colors assigned to the neighbors of $v$ in the configuration $X_0$. By the transition rule of the Markov chain, the value $X_0(v)$ is then sampled uniformly from $[q] \setminus S$.

Under the inductive coupling at times $t < 0$, the oracle values $\+O(u)$ returned in any call to $\chk(0, c)$ or $\resolve(0)$ are exactly $X_0(v_{i(\pred_u(0))})$. Thus, we can extend the coupling to $t = 0$ as follows. 

Maintain a local list $L'$, initially equal to $[q]$, and ensure that $L'$ and the global list $L(0)$ used in the algorithm remain synchronized throughout (meaning the invariant $L'=L(0)$ always holds):
 \begin{itemize}
    \item Whenever $\resolve(0)$ is called:
    \begin{itemize}
        \item If $M(0)\neq \perp$, the output is deterministic and no further randomness needs to be coupled;
        \item Otherwise, by $q\geq 50\Delta$ and Lines~\ref{Line:check-list-size}–\ref{Line:check-resolve-instead} of $\chk(t,c)$, we have $\abs{L(0)}\geq 50\Delta$. By \Cref{item:coloring-local-correctness-1} of \Cref{lemma:coloring-local-correctness}, we can couple the randomness so that the outcome of $\resolve(0)$ is exactly $X_{0}(v_{i(\pred_u(0))})$.
    \end{itemize}
     \item Whenever $\chk(0,c)$ is called for some $c\in [q]$:
     \begin{itemize}
         \item If $M(0)\neq \perp$ or $c\notin L(0)$, the output is deterministic and no further randomness needs to be coupled;
         \item Otherwise, if $\abs{L(0)}\leq 50\Delta$, $\resolve(0)$ is called instead, and this reduces to the case we already proved;
         \item Otherwise we have $\abs{L(0)}\geq 50\Delta$. By \Cref{item:coloring-local-correctness-2} of \Cref{lemma:coloring-local-correctness}, we can couple the randomness so that:
         \begin{itemize}
             \item If $\chk(0,c)$ returns 1, let $X_{0}(v_{i(\pred_u(0))})=c$ (both happens with probability $\frac{1}{|L'\setminus S|}=\frac{1}{|L(0)\setminus S|}$);
             \item Otherwise, remove $c$ from $L'$ (both happens with probability $1-\frac{1}{|L'\setminus S|}=1-\frac{1}{|L(0)\setminus S|}$). As $c$ is also removed from $L(0)$ in this case, $L'$ and $L(0)$ stay the same.
         \end{itemize}
     \end{itemize}
 \end{itemize}
 This completes the construction of the coupling, thereby proving the lemma.
\end{proof}

Now, we are finally ready to prove \Cref{lemma:partial-rejection-conditional-correctness}.
\begin{proof}[Proof of \Cref{lemma:partial-rejection-conditional-correctness}]
For any $T\ge 0$, we couple the randomness of $\lsample(\Lambda)$ with $\lsample_T(\Lambda)$ by pre-sampling all random variables used  within $\resolve(t)$ and $\chk(t,c)$ for each $t\leq 0$. Here, the coupling fails if and only if recursion reaches some timestamp $t'\leq -T$. Note that when $\lsample_T(\Lambda)$ terminates almost surely, the probability of any recursion within $\lsample_T(\Lambda)$ reaching some timestamp $t'\leq -T$ must tend to $0$ as $T$ tends to infinity, meaning the coupling above gives us
\begin{equation}\label{eq:dtv-lsample-converge}
    \lim\limits_{T\to \infty}\DTV{\lsample(\Lambda)}{\lsample_T(\Lambda)}=0.
\end{equation}

Note that when $q\geq \Delta+1$, we have $\+P(T)$ is irreducible and \Cref{thm-convergence} implies that
\begin{equation}\label{eq:dtv-convergence-irreducibility-2}
\lim_{T\to \infty}\DTV{\mu_{\Lambda}}{X_{T,0}(\Lambda)}=0,
\end{equation}
Therefore, 
 \begin{align*}
  &\DTV{\mu_{\Lambda}}{\lsample(\Lambda)} \\
  \le~&~\limsup_{T\rightarrow\infty}\DTV{\mu_{\Lambda}}{\lsample_T(\Lambda)} \tag{by triangle inequality}\\
  +~&~ \limsup_{T\rightarrow\infty}\DTV{\lsample_T(\Lambda)}{\lsample(\Lambda)}  \\
  =~&~\limsup_{T\rightarrow\infty}\DTV{\mu_{\Lambda}}{\lsample_T(\Lambda)} \tag{by \eqref{eq:dtv-lsample-converge}} \\
  \le~&~\limsup_{T\rightarrow\infty}\DTV{\mu_{\Lambda}}{X_{T,0}(\Lambda)} + \limsup_{T\rightarrow\infty}\DTV{X_{T,0}(\Lambda)}{\lsample_T(\Lambda)}   \tag{by triangle inequality}\\
   =~&~\limsup_{T\rightarrow\infty}\DTV{\mu_{\Lambda}}{X_{T,0}(\Lambda)}   \tag{by \Cref{lemma:lsample-coloring-identical}}\\
   =~&~0, \tag{by \eqref{eq:dtv-convergence-irreducibility-2}}
\end{align*}
concluding the proof of the lemma.
\end{proof}

\subsection{Efficiency of the local sampler}

Now, we proceed to proving the efficiency of our local sampler for $q$-colorings. Our proof is done by designing a carefully-chosen potential function that relates to the state of the algorithm, and showing that such a potential function decays in expectation at each step as the algorithm evolves.

\begin{lemma}[efficiency of the local sampler for $q$-colorings]\label{lemma:q-coloring-efficiency}
 Suppose that the input $q$-coloring instance $G=(V,E)$ satisfies $q\geq 65\Delta$ where $\Delta\geq 1$ is the maximum degree of $G$. Then for any subset $\Lambda\subseteq V$, $\lsample(\Lambda)$ terminates almost surely within expected time $|\Lambda|\cdot \Delta^2 q$.
\end{lemma}
\begin{proof}

Following the proof of \Cref{lemma:lsample-coloring-identical}, we can couple the randomness of the algorithm and $(X_t)_{t \leq 0}$ in such a way that the following two invariants hold:

\begin{enumerate}
 \item Whenever $\chk(t, c)$ is called for some timestamp $t \leq 0$, some color $c \in [q]$ and \emph{terminates}, the outcome is $1$ if and only if $X_t(v_{i(t)}) = c$; 
\item Whenever $\resolve(t)$ is called for some timestamp $t \leq 0$ and \emph{terminates}, its output is $X_t(v_{i(t)})$. 
\end{enumerate}

Note that here $(X_t)_{t \leq 0}$ is viewed as the limiting process of $(X_t)_{-T\leq t\leq 0}$ as $T\to \infty$. When $q\geq \Delta+1$, the Markov chain $\+P(T)$ is ergodic, and the distribution specified by $(X_t)_{t \leq 0}$ projecting onto any finite subset of timestamps is always well-defined.

We then strengthen the lemma and prove that under any realization of the outcome of updates in $(X_t)_{t \leq 0}$, the desired result that for any $\Lambda\subseteq V$, $\lsample(\Lambda)$ terminates almost surely within expected time $|\Lambda|\cdot \Delta^2 q$ always holds.

Note that the state $\Pi$ of the algorithm $\lsample(\Lambda)$ can be defined by:
\begin{itemize}
\item the state of the global data structures $M$ and $L$, and
\item the current stack of recursive calls to $\resolve(t)$ and $\chk(t,c)$.
\end{itemize}

We introduce a potential function to track the state of the algorithm.

\begin{definition}[potential function associated with the state of the algorithm]\label{definition:potential-function}
Let $\Phi = \Phi(\Pi) \in \mathbb{R}$ be the potential function, defined as follows. Initially, $\Phi = 0$.
\begin{itemize}
\item For each \emph{distinct} $\resolve(t)$ call such that $M(t)=\perp$ in the current stack of the algorithm, update $\Phi \gets \Phi + C_1\Delta$.
\item For each \emph{distinct} $\chk(t,c)$ call such that $M(t)=\perp$ and $c\in L(t)$ in the current stack of the algorithm, update $\Phi \gets \Phi + C_2$.
\item For each $t \in \mathbb{Z}$ such that $M(t) = \perp$ and \emph{there is currently no $\resolve(t)$ call in the stack}, update $\Phi \gets \Phi + C_3 \cdot (65\Delta - |L|)$; that is, add $C_3$ to the potential function for each removed color in the list of an yet undetermined timestamp.
\end{itemize}
Here, $C_1,C_2,C_3$ are universal constants to be determined later in this proof.
\end{definition}

At the start of the algorithm, for each $v \in \Lambda$, the only call is $\resolve(\pred_0(v))$. Thus, the potential function is initially set to $\Phi = |\Lambda| \cdot C_1 \Delta$, as given by \Cref{definition:potential-function}. We will show that $\Phi$ contracts in expectation throughout the algorithm's execution.

 Without loss of generality, we only consider \emph{active} calls of the form:
\begin{itemize}
    \item $\resolve(t)$ such that $M(t)=\perp$;
    \item $\chk(t,c)$ such that $M(t)=\perp$ and $c\in L(t)$,
\end{itemize}
as the remaining calls directly terminate by memoization, and we attribute the running time of such calls to the procedure that incurred them. The algorithm terminates when there are no active $\resolve(t)$ or $\chk(t,c)$ calls remaining in the stack. 

We also decompose the execution of the algorithm into \emph{steps}, where in each step, we reveal certain random choices made by the algorithm, and the order of these random choices may differ from the execution order of the algorithm. We analyze the expected change of $\Phi$ in several cases.
\begin{itemize}
\item Suppose the top of the stack is a $\resolve(t)$ call for some $t \in \mathbb{Z}$ such that $M(t)=\perp$. This call proceeds by calling $\eval^{\+O}(v_{i(t)};L(t))$. We define each iteration of the \textbf{while} loop in $\eval^{\+O}(v_{i(t)};L(t))$ as a single \emph{step}, and we track the expected change in $\Phi$ after each step.

   Since $|L| = |L(t)| \geq 50\Delta$ at the beginning, and $|S_t|$ has size at most $\Delta$, we always have $|L| \geq 49\Delta$ at the start of any iteration. Therefore, the condition in \Cref{Line:evaluate-coloring-check} of $\eval^{\+O}(v_{i(t)};L(t))$ is satisfied with probability at least $\frac{49}{50}$, regardless of the realization of $S_t$. This allows the number of iterations to be stochastically dominated by $\mathrm{Geo}(\frac{49}{50})$.

Hence, in each step:
\begin{itemize}
    \item $\Delta$ calls to $\chk(t,c)$ are added to the stack, increasing the potential function by $C_2\Delta $.
    \item With probability at least $\frac{49}{50}$, the $\resolve(t)$ call is removed from the stack, decreasing the potential function by $C_1\Delta$.
\end{itemize}
We remark that during these steps, the potential function is not affected by the change in the size of $L(t)$ as $\resolve(t)$ is already in the stack of the algorithm.

Therefore, the expected change in the potential function per step is at most:
\[
C_2\Delta - \frac{49}{50} \cdot C_1\Delta .
\]

\item Suppose the top of the stack is a $\chk(t,c)$ call for some $t \in \mathbb{Z}$ and $c \in [q]$ such that $M(t)= \perp$ and $c \in L(t)$. 
\begin{itemize}
\item If $|L(t)| = 50\Delta$, then the $\chk(t,c)$ call is removed from the stack, decreasing the potential function by $C_2$.
\begin{itemize}
    \item If there currently is a $\resolve(t)$ call in the stack, then the total change in the potential function is simply $-C_2$.
    \item Otherwise, a new $\resolve(t)$ call is added to the stack, increasing the potential function by $C_1\Delta$. Additionally, since the removed color in the list $L(t)$ no longer contributes to the potential function when this new $\resolve(t)$ call is added, the potential function further decreases by $C_3 \cdot (65\Delta - 50\Delta) = 15\Delta \cdot C_3$. Therefore, the total change in the potential function is:
\[
C_1\Delta - C_2 - 15\Delta \cdot C_3.
\]
\end{itemize}
\item Otherwise, we assume $M(t)\neq \perp, c\in L(t)$ and that $\abs{L(t)}>50\Delta$. This call proceeds by calling $\eval^{\+O}(v_{i(t)},c;L(t))$. According to $\eval^{\+O}(v_{i(t)},c;L(t))$:
\begin{itemize}
    \item With probability $1-\frac{2}{\abs{L}}\geq 1-\frac{1}{25\Delta}$, $\eval^{\+O}(v_{i(t)},c;L(t))$ directly terminates and returns $0$. In this case, the $\chk(t,c)$ call is removed from the stack, decreasing the potential function by $C_2$. Also, due to \Cref{Line:check-list-remove} of $\chk(t,c)$, $c$ is removed from $L(t)$, increasing the potential function by at most $C_3$. Hence, in this case, the change in the potential function is at most $C_3-C_2$.
    \item With probability $\frac{2}{\abs{L}}\leq \frac{1}{25\Delta}$, $\eval^{\+O}(v_{i(t)},c;L(t))$ first adds at most $\Delta$ recursive calls of the form $\chk(t',c')$, then uses a Bernoulli factory algorithm for division to return $1$ with probability $\frac{|L|/2}{\abs{L\setminus S_t}}$, accessing \Cref{Alg:access-p-coin} as an ${|L\setminus S_{t}|}/{|L|}$-probability coin, where $S_t$ is defined in \eqref{eq:definition-St}. Note that after fixing the realization of $(X_t)_{t\leq 0}$, the outcome of \Cref{Alg:access-p-coin} is solely determined by \Cref{Line:access-p-coin-sample}, which are clearly independent of the randomness within recursive calls to \Cref{Line:access-p-coin-sample}. Therefore, we can first determine the randomness within this Bernoulli factory algorithm and consider this as a \emph{step}. 
    
    According to \Cref{lemma:correctness-Bernoulli-factory}, the expected number of calls to \Cref{Alg:access-p-coin} is at most $21$. Each such call may in turn generate up to $\Delta$ recursive calls of the form $\chk(t',c')$, with each recursive call contributing at most $C_2$ to the potential function. At the same time, the original $\chk(t,c)$ call is removed from the stack, decreasing the potential function by $C_2$; and, due to \Cref{Line:check-list-remove}, the color $c$ is removed from $L(t)$, increasing the potential function by $C_3$. Thus, the total expected change in the potential function in this case is at most:
    \[
    (22\Delta-1)C_2+C_3.
    \]
   Moreover, since at most $\Delta q$ distinct recursive calls of the form $\chk(t',c')$ can be generated where $t'=\pred_u(t)$ for some $u\in N(v_{i(t)})$, the maximum change in the potential function is at most:
    \[
    (\Delta q-1) C_2+C_3.
    \]
\end{itemize}
Summarizing, in this case, the expected change in the potential function per step is at most:
\[
\left(1-\frac{1}{25\Delta}\right)\cdot (C_3-C_2)+\frac{1}{25\Delta}\cdot \left((22\Delta-1)C_2+C_3\right)=C_3-\frac{3C_2}{25}.
\]
\end{itemize}
\end{itemize}
We then set the constants as follows:
\[
C_1=30, \quad C_2=25, \quad C_3=2.
\]
In this case, it is easy to see that the expected change of the potential function after each step is at most $-1$.
Let $\Phi_0,\Phi_1,\dots$ be the random sequence of the potential function after the $i$-th step of the algorithm. Define $\Psi_i=\Phi_i+i$ for each $i\geq 0$. By the previous analysis, we have $\{\Psi_i\}_{i\geq 0}$ forms a supermartingale with each absolute increment $|\Psi_{i+1}-\Psi_{i}|$ upper bounded by some $K=\poly(q,\Delta)$. Let $\tau=\min\{i>0\mid \Phi_i=0\}=\min\{i>0\mid \Psi_i=i\}$ be a stopping time. By the Azuma-Hoeffding inequality, we have for each fixed $\varepsilon>0$,
\[
\Pr{\tau>(1+\varepsilon) 30|\Lambda|\cdot \Delta}\leq \Pr{\Psi_{(1+\varepsilon)30|\Lambda|\cdot \Delta}>\Psi_0+30\varepsilon|\Lambda|\cdot \Delta}\leq \mathrm{exp}\left(\frac{-30\varepsilon^2|\Lambda|\Delta}{(1+\varepsilon) K^2}\right),
\]
from which we can conclude that $\E{\tau}<\infty$, allowing us to apply Doob's optional stopping theorem:
\[
30|\Lambda|\cdot \Delta=\E{\Psi_0}\geq \E{\Psi_\tau}=\E{\tau},
\]
meaning that the algorithm executes at most $30|\Lambda|\cdot \Delta$ steps in expectation. Consequently, by \Cref{lemma:correctness-Bernoulli-factory}, the expected running time of the algorithm is bounded by $O(\abs{\Lambda}\cdot \Delta q^2)$.
\end{proof}

Lemmas \ref{lemma:coloring-local-correctness} and \ref{lemma:q-coloring-efficiency} together establish \Cref{theorem:coloring}. For \Cref{theorem:inference-coloring}, observe that all algorithms and proofs in this section naturally extend to the setting of list colorings, where each vertex $v \in V$ has its own color list $Q_v$ satisfying $\abs{Q_v} \geq 65\Delta$. By applying the same argument as in the proof of \Cref{theorem:inference-permissive}, \Cref{theorem:inference-coloring} follows.

\section{Conclusions and Open Problems}\label{sec:conclusions}
In this paper, we design new local samplers that go beyond the use of the local uniformity property by generalizing and refining the framework of backward deduction of Markov chains, i.e., the ``Coupling Towards The Past'' (CTTP) method. 
Specifically, we design the first local samplers for both spin systems with soft constraints in near-critical regimes, and uniform $q$-coloring under the near-critical condition of $q=O(\Delta)$ where $\Delta$ is the maximum degree of the graph. 
The proposed local samplers are perfect, achieve near-linear runtime, and admit direct applications to local algorithms for probabilistic inference within the same parameter regimes.

We leave the following open problems and directions for future work: 
\begin{itemize}
    \item While our local sampler performs well in near-critical regimes for spin systems using backward deduction of Glauber dynamics, it has been shown that forward simulation of Glauber dynamics mixes rapidly for the Ising model up to the uniqueness threshold. Can we improve the analysis of our current algorithm, design new local samplers that are efficient up to this critical threshold, or prove a lower bound showing that this is intractable for local samplers?
    \item 
    Both the CTTP (Coupling Towards The Past) and the CFTP (Coupling From The Past) approaches are based on grand coupling and can yield perfect sampling algorithms.
    However, the local sampler via CTTP requires a more restrictive grand coupling that admits a local implementation. 
    For $q$-colorings, we show that CTTP with such a local grand coupling can indeed achieve the $q=O(\Delta)$ threshold; in particular, we establish $q\ge 65\Delta$. 
    In contrast, the best known result for CFTP with a global grand coupling is  $q \geq (\frac{8}{3} + o(1))\Delta$, as achieved in~\cite{jain2021perfectly}, which attains a better constant factor.
    This raises a natural question: can one obtain a local sampler that matches or even surpasses the global grand coupling threshold?
    \item Local samplers, as introduced in~\cite{AJ22,feng2023towards}, have found a wide range of applications, as discussed in the introduction. Our work further shows that local samplers can also imply efficient local algorithms for probabilistic inference. We hope to see more applications of local samplers, particularly in the design of distributed and parallel algorithms. 
\end{itemize}

\ifarxiv{
\section*{Acknowledgement}
Chunyang would like to thank Jingcheng Liu and Yixiao Yu for pointing out Lubetzky and Sly's result, which analyzes the cutoff phenomenon of the Ising model in near-critical regimes.
}

\newpage

\bibliographystyle{alpha}
\bibliography{references} 

\newcommand{\etalchar}[1]{$^{#1}$}
\begin{thebibliography}{FGW{\etalchar{+}}23}

\bibitem[AFF{\etalchar{+}}24]{anand24approximate}
Konrad Anand, Weiming Feng, Graham Freifeld, Heng Guo, and Jiaheng Wang.
\newblock {Approximate Counting for Spin Systems in Sub-Quadratic Time}.
\newblock In {\em ICALP}, volume 297, pages 11:1--11:20. Schloss Dagstuhl -- Leibniz-Zentrum f{\"u}r Informatik, 2024.

\bibitem[AFG{\etalchar{+}}25]{anand2025sinkfree}
Konrad Anand, Graham Freifeld, Heng Guo, Chunyang Wang, and Jiaheng Wang.
\newblock Sink-free orientations: a local sampler with applications, 2025.
\newblock to appear in RANDOM 2025.

\bibitem[AGPP23]{anand23sphere}
Konrad Anand, Andreas G{\"{o}}bel, Marcus Pappik, and Will Perkins.
\newblock Perfect sampling for hard spheres from strong spatial mixing.
\newblock In {\em RANDOM}, volume 275 of {\em LIPIcs}, pages 38:1--38:18. Schloss Dagstuhl - Leibniz-Zentrum f{\"{u}}r Informatik, 2023.

\bibitem[AJ22]{AJ22}
Konrad Anand and Mark Jerrum.
\newblock Perfect sampling in infinite spin systems via strong spatial mixing.
\newblock {\em SIAM Journal on Computing}, 51(4):1280--1295, 2022.

\bibitem[AJK{\etalchar{+}}22]{anari2022entropic}
Nima Anari, Vishesh Jain, Frederic Koehler, Huy~Tuan Pham, and Thuy{-}Duong Vuong.
\newblock Entropic independence: optimal mixing of down-up random walks.
\newblock In {\em {STOC}}, pages 1418--1430. {ACM}, 2022.

\bibitem[ALO20]{ALO20}
Nima Anari, Kuikui Liu, and Shayan {Oveis Gharan}.
\newblock Spectral independence in high-dimensional expanders and applications to the hardcore model.
\newblock In {\em {FOCS}}, pages 1319--1330. {IEEE}, 2020.

\bibitem[BC20]{bhandari2020improved}
Siddharth Bhandari and Sayantan Chakraborty.
\newblock Improved bounds for perfect sampling of k-colorings in graphs.
\newblock In {\em STOC}, pages 631--642, 2020.

\bibitem[BPR22]{biswas2022local}
Amartya~Shankha Biswas, Edward Pyne, and Ronitt Rubinfeld.
\newblock {Local Access to Random Walks}.
\newblock In {\em ITCS}, volume 215, pages 24:1--24:22. Schloss Dagstuhl -- Leibniz-Zentrum f{\"u}r Informatik, 2022.

\bibitem[BRY20]{amartya2020local}
Amartya~Shankha Biswas, Ronitt Rubinfeld, and Anak Yodpinyanee.
\newblock Local access to huge random objects through partial sampling.
\newblock In {\em ITCS}. Schloss Dagstuhl-Leibniz-Zentrum f{\"u}r Informatik, 2020.

\bibitem[CDM{\etalchar{+}}19]{chen2019improved}
Sitan Chen, Michelle Delcourt, Ankur Moitra, Guillem Perarnau, and Luke Postle.
\newblock Improved bounds for randomly sampling colorings via linear programming.
\newblock In {\em SODA}, page 2216–2234, USA, 2019. SIAM.

\bibitem[CE22]{CE22}
Yuansi Chen and Ronen Eldan.
\newblock Localization schemes: A framework for proving mixing bounds for markov chains (extended abstract).
\newblock In {\em FOCS}, pages 110--122. {IEEE}, 2022.

\bibitem[CFYZ21]{chen2021rapid}
Xiaoyu Chen, Weiming Feng, Yitong Yin, and Xinyuan Zhang.
\newblock Rapid mixing of {G}lauber dynamics via spectral independence for all degrees.
\newblock In {\em FOCS}, pages 137--148, 2021.

\bibitem[CFYZ22]{CFYZ22}
Xiaoyu Chen, Weiming Feng, Yitong Yin, and Xinyuan Zhang.
\newblock Optimal mixing for two-state anti-ferromagnetic spin systems.
\newblock In {\em FOCS}, pages 588--599. {IEEE}, 2022.

\bibitem[CLV20]{chen2020rapid}
Zongchen Chen, Kuikui Liu, and Eric Vigoda.
\newblock Rapid mixing of {G}lauber dynamics up to uniqueness via contraction.
\newblock In {\em FOCS}, pages 1307--1318. IEEE, 2020.

\bibitem[CLV21]{CLV21}
Zongchen Chen, Kuikui Liu, and Eric Vigoda.
\newblock Optimal mixing of {G}lauber dynamics: entropy factorization via high-dimensional expansion.
\newblock In {\em {STOC}}, pages 1537--1550. {ACM}, 2021.

\bibitem[CV25]{carlson2025flip}
Charlie Carlson and Eric Vigoda.
\newblock Flip dynamics for sampling colorings: Improving (11/6 — {$\varepsilon$}) using a simple metric.
\newblock In {\em SODA}, pages 2194--2212, 2025.

\bibitem[DL93]{dagum1993approximate}
Paul Dagum and Michael Luby.
\newblock Approximating probabilistic inference in bayesian belief networks is np-hard.
\newblock {\em Artif. Intell.}, 60(1):141–153, 1993.

\bibitem[DL97]{dagum1997optimal}
Paul Dagum and Michael Luby.
\newblock An optimal approximation algorithm for bayesian inference.
\newblock {\em Artif. Intell.}, 93(1–2):1–27, 1997.

\bibitem[FGW{\etalchar{+}}23]{feng2023towards}
Weiming Feng, Heng Guo, Chunyang Wang, Jiaheng Wang, and Yitong Yin.
\newblock Towards derandomising {M}arkov chain {M}onte {C}arlo.
\newblock In {\em FOCS}, pages 1963--1990. IEEE, 2023.

\bibitem[G{\v{S}}V16]{galanis2016inapproximability}
Andreas Galanis, Daniel {\v{S}}tefankovi{\v{c}}, and Eric Vigoda.
\newblock Inapproximability of the partition function for the antiferromagnetic {Ising} and hard-core models.
\newblock {\em Combinatorics, Probability and Computing}, 25(04):500--559, 2016.

\bibitem[Hub98]{huber1998exact}
Mark Huber.
\newblock Exact sampling and approximate counting techniques.
\newblock In {\em STOC}, pages 31--40, 1998.

\bibitem[Hub16]{Hub16Bernoulli}
Mark Huber.
\newblock Nearly optimal {B}ernoulli factories for linear functions.
\newblock {\em Combin. Probab. Comput.}, 25(4):577--591, 2016.

\bibitem[HWY22]{HWY22a}
Kun He, Chunyang Wang, and Yitong Yin.
\newblock Sampling lovász local lemma for general constraint satisfaction solutions in near-linear time.
\newblock In {\em FOCS}, pages 147--158. IEEE, 2022.

\bibitem[HWY23]{HWY22c}
Kun He, Chunyang Wang, and Yitong Yin.
\newblock Deterministic counting lovász local lemma beyond linear programming.
\newblock In {\em SODA}, pages 3388--3425. SIAM, 2023.

\bibitem[Isi25]{Ising1925Beitrag}
Ernst Ising.
\newblock Beitrag zur theorie des ferromagnetismus.
\newblock {\em Zeitschrift f{\"u}r Physik}, 31:253--258, 1925.

\bibitem[Jer95]{jerrum1995simple}
Mark Jerrum.
\newblock A very simple algorithm for estimating the number of k‐colorings of a low‐degree graph.
\newblock {\em Random Struct. Algorithms}, 7(2):157–165, September 1995.

\bibitem[JSS21]{jain2021perfectly}
Vishesh Jain, Ashwin Sah, and Mehtaab Sawhney.
\newblock Perfectly sampling k>(8/3+o(1)) $\delta$-colorings in graphs.
\newblock In {\em STOC}, pages 1589--1600, 2021.

\bibitem[LPW17]{levin2017markov}
David~A. Levin, Yuval Peres, and Elizabeth~L. Wilmer.
\newblock {\em Markov chains and mixing times}.
\newblock American Mathematical Society, Providence, RI, 2017.

\bibitem[LS16]{Lubetzky2016information}
Eyal Lubetzky and Allan Sly.
\newblock Information percolation and cutoff for the stochastic {I}sing model.
\newblock {\em J. Amer. Math. Soc.}, 29(3):729--774, 2016.

\bibitem[LS17]{Lubetzky2017universal}
Eyal Lubetzky and Allan Sly.
\newblock {Universality of cutoff for the Ising model}.
\newblock {\em The Annals of Probability}, 45(6A):3664 -- 3696, 2017.

\bibitem[Mor21]{morina2021extending}
G.~Morina.
\newblock {\em Extending the {Bernoulli} {Factory} to a {Dice} {Enterprise}}.
\newblock University of Warwick, 2021.

\bibitem[MSW22]{morters2022sublinear}
Peter M\"{o}rters, Christian Sohler, and Stefan Walzer.
\newblock {A Sublinear Local Access Implementation for the Chinese Restaurant Process}.
\newblock In {\em RANDOM}, volume 245, pages 28:1--28:18. Schloss Dagstuhl -- Leibniz-Zentrum f{\"u}r Informatik, 2022.

\bibitem[NP05]{Nacu2005FastSO}
\c{S}erban Nacu and Yuval Peres.
\newblock Fast simulation of new coins from old.
\newblock {\em Ann. Appl. Probab.}, 15(1A):93--115, 2005.

\bibitem[PW96]{propp1996exact}
James~Gary Propp and David~Bruce Wilson.
\newblock Exact sampling with coupled markov chains and applications to statistical mechanics.
\newblock {\em Random Structures \& Algorithms}, 9(1-2):223--252, 1996.

\bibitem[SS14]{SS14}
Allan Sly and Nike Sun.
\newblock Counting in two-spin models on {$d$}-regular graphs.
\newblock {\em Ann. Probab.}, 42(6):2383--2416, 2014.

\bibitem[Vig99]{vigoda1999improved}
Eric Vigoda.
\newblock Improved bounds for sampling colorings.
\newblock In {\em FOCS}, pages 51--59, 1999.

\bibitem[VN51]{vonNeumann1951}
John Von~Neumann.
\newblock various techniques used in connection with random digits.
\newblock {\em Appl. Math Ser}, 12(36-38):3, 1951.

\end{thebibliography}

\newpage
\appendix
\section{Bernoulli factory method for simulating probability}\label{sec:bernoulli-factory}

In this section, we explicitly provide the construction of the Bernoulli factory for division, and formally prove \Cref{lemma:correctness-Bernoulli-factory}.

For $\xi\in [0,1]$ we denote by $\+O_{\xi}$ a coin (or oracle) that returns 
$1$ (heads) with probability $\xi$, and $0$ (tails) with probability $1-\xi$, independently on each call.

Our construction of the Bernoulli factory for division is based on the method from~\cite{morina2021extending}, which in turn builds on the subtraction Bernoulli factory from~\cite{Nacu2005FastSO}. That construction is itself based on a linear Bernoulli factory, for which we adopt the implementation from~\cite{Hub16Bernoulli}.

We first introduce the linear Bernoulli factory, which transforms $\+{O}_\xi$ to $\+{O}_{C\xi}$ for a $C>1$ with the promise that $C\xi\leq 1$.
We adopt the construction of linear Bernoulli factory in \cite{Hub16Bernoulli} described in \Cref{Alg:lbf}.
Its correctness and efficiency is guaranteed as follows.

\begin{algorithm}
  \caption{$\linbf{}(\+{O}, C, \zeta)$\cite{Hub16Bernoulli}} \label{Alg:lbf}
  \KwIn{a coin $\+{O}=\+{O}_{\xi}$ with unknown $\xi$,  $C>1$ and a slack $\zeta>0$, with promise that $C\xi\leq 1-\zeta$;}  
  \KwOut{a random value from $\+O_{C\xi}$;}
  $k\gets 4.6/\zeta,\zeta\gets\min\{\zeta,0.644\},i\gets 1$\;
  \Repeat{$i=0$ or $R=0$}
  {
    \Repeat{$i=0$ or $i\geq k$}
    {
        draw $B\gets \+{O}$, $G\gets \textsf{Geo}\left(\frac{C-1}{C}\right)$\;\label{Line:linear-bf-draw}
        \tcp{\small$G$ is drawn according to geometric distribution with parameter $\frac{C-1}{C}$}
        $i\gets i-1+(1-B)G$\;\label{Line:linear-bf-update}
    }
    \If{$i\geq k$}
    {
        draw $R\gets \textsf{Bernoulli}\left((1+\zeta/2)^{-i}\right)$\;
        $C\gets C(1+\zeta/2),\zeta\gets \zeta/2,k\gets 2k$\label{Line:linear-bf-update-C}\;
    }
  }
    \textbf{return} $\one{i=0}$\;
\end{algorithm}

\begin{proposition}[\text{\cite[Theorem 1]{Hub16Bernoulli}}]\label{linbfcor}
Given access to a coin $\+{O}_{\xi}$, given as input $C>1$ and $\zeta>0$, with promise that $C\xi\leq 1-\zeta$, 
$\linbf{}(\+{O}, C, \zeta)$ terminates with probability $1$ and returns a draw of $\+{O}_{C\xi}$. Moreover, the expected number of calls to $\+O_{\xi}$ is at most $9.5C\zeta^{-1}$.
\end{proposition}



Building upon the linear Bernoulli factory, a subtraction Bernoulli factory $\subbf{}(\+{O}_{\xi_1}, \+{O}_{\xi_2}, \zeta)$ was proposed in~\cite{Nacu2005FastSO}. It transforms two coins $\+{O}_{\xi_1}$ and $\+{O}_{\xi_2}$—under the promise that $\xi_1 - \xi_2 \geq \zeta > 0$—into a new coin $\+{O}_{\xi_1 - \xi_2}$.

We implement this procedure using the linear Bernoulli factory defined above:

\begin{itemize}
\item $\subbf{}\left(\+{O}_{\xi_1}, \+{O}_{\xi_2}, \zeta\right) = 1- \linbf{}\left(\+{O}_{(1 - \xi_1 + \xi_2)/2}, 2, \zeta\right)$,
\end{itemize}

where the coin $\+{O}_{(1 - \xi_1 + \xi_2)/2}$ is realized using $\+{O}_{1/2}$, $\+{O}_{\xi_1}$, and $\+{O}_{\xi_2}$ as follows: if $\+{O}_{1/2} = 1$, return $1- \+{O}_{\xi_1}$; otherwise, return $\+{O}_{\xi_2}$.

The correctness and efficiency of this procedure is guaranteed as follows.

\begin{corollary}[\text{c.f.~\cite[Proposition 14, (iv)]{Nacu2005FastSO}}]\label{subbfcor}
Given access to two coins $\+{O}_{\xi_1}$ and $\+{O}_{\xi_2}$, and given as input $\zeta>0$, with promise that $\xi_1-\xi_2\geq \zeta$, 
$\subbf{}\left(\+{O}_{\xi_1}, \+{O}_{\xi_2},\zeta\right)$ terminates with probability $1$ and returns a draw of $\+{O}_{\xi_1-\xi_2}$. Moreover, the expected number of call to $\+O_{\xi}$ and $\+O_{\xi_2}$ is at most $9.5\zeta^{-1}$ each.
\end{corollary}

\sloppy Now we can formally give the construction of the Bernoulli factory for division in~\cite{morina2021extending}, $\berdiv(\+O_{\xi},p,\zeta)$, which transforms a coin $\+O_{\xi}$  with the promise that $\xi-p\geq \zeta$ into a new coin $\+O_{p/\zeta}$. The construction is given as follows~\cite[Algorithm 9]{morina2021extending}.

Repeat the following until a value is returned:
\begin{itemize}
    \item If a draw $I$ from $\+O_{1/2}$ is $1$: return $1$ with probability $p$. 
    \item Otherwise, return $0$ if a draw from $\+O_{\xi-p}$ returns $1$,
\end{itemize}
where $\+O_{\xi-p}$ is implemented using $\subbf\left(\+{O}_{\xi}, \+{O}_{p},\zeta\right)$. 

The correctness and efficiency of this procedure is guaranteed as follows, which directly proves \Cref{lemma:correctness-Bernoulli-factory}.
\begin{corollary}[\text{c.f.~\cite[Proposition 2.24]{morina2021extending}}]\label{bfdivcor}
Given access to a coin $\+{O}_{\xi}$, and given as input $0\leq p\leq 1,\zeta>0$, with promise that $\xi-p\geq \zeta$, 
$\berdiv{}\left(\+{O}_{\xi}, p,\zeta\right)$ terminates with probability $1$ and returns a draw of $\+{O}_{p/\xi}$. Moreover, the expected number of calls to $\+O_{\xi}$ is at most $9.5\xi^{-1}\zeta^{-1}$.
\end{corollary}

\clearpage

\end{document}